\documentclass[reqno, a4paper]{amsart}
\usepackage[utf8]{inputenc}
\usepackage[T1]{fontenc}
\usepackage{amsthm}
\usepackage{amsmath}
\usepackage{multicol}
\usepackage{multirow,bigdelim}
\usepackage[utf8]{inputenc}
\usepackage{lmodern}

\usepackage{tikz} 
\tikzset{
main node/.style={inner sep=0,outer sep=0},
label node/.style={inner sep=0,outer ysep=.2em,outer xsep=.4em,font=\scriptsize,overlay},
strike out/.style={shorten <=-.2em,shorten >=-.5em,overlay}
}

\makeatletter
\ifcase \@ptsize \relax
  \newcommand{\miniscule}{\@setfontsize\miniscule{4}{5}}%
\or
\or
  \newcommand{\miniscule}{\@setfontsize\miniscule{5}{6}}%
\fi
\makeatother

\makeatletter
\ifcase \@ptsize \relax
  \newcommand{\nano}{\@setfontsize\miniscule{3.5}{4.5}}%
\or
  \newcommand{\nano}{\@setfontsize\miniscule{4.5}{5.5}}%
\or
  \newcommand{\nano}{\@setfontsize\miniscule{4.5}{5.5}}%
\fi
\makeatother

\newcommand{\balita}{\raisebox{1.8pt}{\text{ \nano$\bullet$\hspace{1.7pt} }}}

\usepackage{mathtools}

\usepackage{textcomp}

\numberwithin{equation}{section}
\newtheoremstyle{mytheoremstyle} 
    {10pt}                    
    {8pt}                    
    {\itshape}                   
    {}                           
    {\scshape}                   
    {.}                          
    {.5em}                       
    {}  

\theoremstyle{mytheoremstyle}

\newtheorem{thm}{Theorem}[section]
 \newtheorem{corollary}[thm]{Corollary}
 \newtheorem{lemma}[thm]{Lemma}

 \newtheorem{proposition}[thm]{Proposition}
 \newtheoremstyle{definition} 
    {8pt}                    
    {5pt}                    
    {}                   
    {}                           
    {\scshape}                   
    {.}                          
    {.5em}                       
    {}  

 \theoremstyle{definition}
 \newtheorem{definition}[thm]{Definition}
 \newtheorem{example}[thm]{Example}
 \newtheorem*{acknowledgements}{Acknowledgements}
 \newtheorem{remark}[thm]{Remark}
 
 \newenvironment{proof*}[1]{\renewcommand{\proofname}{#1}%
 \begin{proof}}
 {\end{proof}}
\usepackage{graphicx}
\usepackage{xcolor}
 \definecolor{VerdeFH}{HTML}{009374}

\newenvironment{salign} 
  {\csname align*\endcsname}
  {\csname endalign*\endcsname} 
  
  \newcommand{\numerada}{\refstepcounter{equation}\tag{\theequation}}

\DeclareFontEncoding{LS1}{}{}
\DeclareFontSubstitution{LS1}{stix}{m}{n}

 \newcommand*\aay{%
  \text{%
  \fontencoding{LS1}%
  \fontfamily{stixscr}%
  \fontseries{\textmathversion}%
  \fontshape{n}%
  \selectfont\symbol{'141}}}
 
 \newcommand*\fay{%
  \text{%
  \fontencoding{LS1}%
  \fontfamily{stixscr}%
  \fontseries{\textmathversion}%
  \fontshape{n}%
  \selectfont\symbol{'146}}}
 
\newcommand*\kay{%
  \text{%
  \fontencoding{LS1}%
  \fontfamily{stixscr}%
  \fontseries{\textmathversion}%
  \fontshape{n}%
  \selectfont\symbol{"6B}}}
  
  \newcommand*\say{%
  \text{%
  \fontencoding{LS1}%
  \fontfamily{stixscr}%
  \fontseries{\textmathversion}%
  \fontshape{n}%
  \selectfont\symbol{'163}}}
  
    \newcommand*\xay{%
  \text{%
  \fontencoding{LS1}%
  \fontfamily{stixscr}%
  \fontseries{\textmathversion}%
  \fontshape{n}%
  \selectfont\symbol{"78}}}
  
  \renewcommand*\day{%
  \text{%
  \fontencoding{LS1}%
  \fontfamily{stixscr}%
  \fontseries{\textmathversion}%
  \fontshape{n}%
  \selectfont\symbol{'144}}}

   \newcommand*\lay{%
  \text{%
  \fontencoding{LS1}%
  \fontfamily{stixscr}%
  \fontseries{\textmathversion}%
  \fontshape{n}%
  \selectfont\symbol{"6C}}}

  \newcommand*\jay{%
  \text{%
  \fontencoding{LS1}%
  \fontfamily{stixscr}%
  \fontseries{\textmathversion}%
  \fontshape{n}%
  \selectfont\symbol{"7C}}}
  
    \newcommand*\iay{%
  \text{%
  \fontencoding{LS1}%
  \fontfamily{stixscr}%
  \fontseries{\textmathversion}%
  \fontshape{n}%
  \selectfont\symbol{"7B}}}

\makeatletter
  \newcommand*\textmathversion{\csname textmv@\math@version\endcsname}
  \newcommand*\textmv@normal{m}
  \newcommand*\textmv@bold{b}
\makeatother

    \usepackage[margin=1.34in]{geometry}
    \usepackage[bookmarks=false]{hyperref}
    \hypersetup{
         colorlinks   = true,
         citecolor    =blue!60!black
    }
    \hypersetup{linkcolor=blue!70!black}

    \numberwithin{equation}{section}
 
\usepackage{bbm}
\usepackage{footmisc}
\newcommand{\tetrah}{\raisebox{-7pt}{\includegraphics[width=5.7pt]{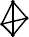}}}

\newcommand{\antiH}{_{\mtr{anti\text{-}Herm.}}}
\newcommand{\sa}{_\mtr{s.a}}

    \renewcommand{\and}{\mbox{and}}

    \newcommand{\mtr}[1]{\mathrm{#1}}
        \newcommand{\mtb}[1]{\mathbb{#1}}
    \newcommand{\mtf}[1]{\mathfrak{#1}}

    \newcommand{\A}{\mathcal{A}}
    \newcommand{\dif}[1]{\mathrm{d}#1}

    \newcommand{\re}{\mathbb{R}}

    \newcommand{\diag}{\mtr{diag}}

    \newcommand{\adj}{\mathrm{ad}}
    \newcommand{\Adj}{\mathrm{Ad}}
    
    \newcommand{\gauge}{\mathsf{G}}

    \renewcommand{\H}{\mathcal{H}}
    
    \newcommand{\C}{\mathbb{C}}

    \usepackage{amssymb}
    \usepackage{mathrsfs}

    \newcommand{\ii}{\mathrm{i}}
    \newcommand{\ee}{\mathrm{e}}
    
    \newcommand{\inv}{^{-1}}
    \newcommand{\mtc}[1]{\mathcal{#1}}
    \newcommand{\Z}{\mathbb{Z}}

    \newcommand{\U}{\mtc{U}}

    \newcommand{\hp}[1]{^{(#1)}}

    \newcommand{\with}{\,\,\mtr{with} \,\,}
    \newcommand{\where}{\mbox{where}\,\,}

    \DeclareMathOperator{\Tr}{Tr}
    \DeclareMathOperator{\im}{im}
    \newcommand{\TrH}{\Tr_\H}
    \newcommand{\TrN}{\Tr_{N}}
 \newcommand{\TrV}{\Tr_{V}}

    \newcommand{\TrMNn}{\Tr_{M_{N\otimes n}^{\C}}}

\newcommand{\YMH}{{\mathrm{YM\text{-}H}}}
\newcommand{\gH}{{\mathrm{g\text{-}H}}}
\newcommand{\Hi}{{\mathrm{H}}}
\newcommand{\YM}{{\mathrm{YM}}}
\newcommand{\Mn}{M_n(\C)}
\newcommand{\MN}{M_N(\C)}
\newcommand{\MNn}{M_{N\otimes n}(\C)}
    \newcommand{\sun}{\mathfrak{su}(n)}
    \newcommand{\suN}{\mathfrak{su}(N)}

    \newcommand{\acomm}[1]{\{ #1,\balita\hspace{1pt}\}}

    \newcommand{\comm}[1]{[ #1,\balita]}
    
    \newcommand{\ac}{\scalebox{0.94}{\ensuremath{\mathsf{a}}}} 
\newcommand{\cc}{\scalebox{0.94}{\ensuremath{\mathsf{c}}}} 
\newcommand{\itemb}{\item[$\balita$]}
\newcommand{\fuz}{_{\fay}\hspace{1pt}}

\newcommand{\eeqref}[1]{Eq. \eqref{#1}}
\newcommand{\sgn}{\mtr{sgn}}
\newcommand{\Dg}{D_{\text{\tiny gauge}}}
\newcommand{\Dh}{D_{\text{\tiny Higgs}}}  
    
\usepackage{booktabs}
\usepackage{colortbl}
\usepackage{amsmath}
\usepackage{xcolor}
\usepackage{graphicx}
\colorlet{tableheadcolor}{gray!19} 
\newcommand{\headcol}{\rowcolor{tableheadcolor}} %
\colorlet{tablerowcolor}{gray!10} 
\newcommand{\topline}{\arrayrulecolor{black}\specialrule{0.1em}{\abovetopsep}{0pt}%
            \arrayrulecolor{tableheadcolor}\specialrule{\belowrulesep}{0pt}{0pt}%
            \arrayrulecolor{black}}
\newcommand{\midline}{\arrayrulecolor{tableheadcolor}\specialrule{\aboverulesep}{0pt}{0pt}%
            \arrayrulecolor{black}\specialrule{\lightrulewidth}{0pt}{0pt}%
            \arrayrulecolor{white}\specialrule{\belowrulesep}{0pt}{0pt}%
            \arrayrulecolor{black}}

\newcommand{\bottomlinec}{\arrayrulecolor{tablerowcolor}\specialrule{\aboverulesep}{0pt}{0pt}%
            \arrayrulecolor{black}\specialrule{\heavyrulewidth}{0pt}{\belowbottomsep}}%

   \let\langleb=\langle
   \let\rangleb=\rangle
\usepackage{mathabx}

\usepackage{bbm}
\newcommand{\xm}{\xay}
\newcommand{\km}{\kay}
\newcommand{\am}{\aay}
\newcommand{\sm}{\say}
\newcommand{\lm}{\lay}

\newcommand{\gamu}{\gamma^\mu}
\newcommand{\gahmu}{\gamma^{\hat\mu}}
\newcommand{\hmu}{{\hat{\mu}}}
\newcommand{\munu}{{\mu\nu}}
\newcommand{\Fcurv}{\mathscr{F}} 
\newcommand{\ganu}{\gamma^\nu}
\newcommand{\gahnu}{\gamma^{\hat\nu}}

\begin{document}

\title[Yang-Mills--Higgs matrix spectral triples]{On multimatrix models motivated
by \\ random noncommutative geometry II:\\[0pt]
A Yang-Mills--Higgs matrix model}


 \author{Carlos I. Perez-Sanchez}
  
  \address{Faculty of Physics, University of Warsaw \newline \indent  
  ul. Pasteura 5, 02-093, Warsaw, Poland, European Union \newline \indent \&  \newline \indent   Institute for Theoretical Physics, University of Heidelberg \newline \indent
  Philosophenweg 19, 69120 Heidelberg, Germany, European Union   \newline \indent
 \hspace{.0cm}
  }
  
 \email{cperez@fuw.edu.pl, perez@thphys.uni-heidelberg}
 
 \thanks{}
 


 
 
  \fontsize{11}{13.09}\selectfont

\begin{abstract}
  We continue the study of fuzzy geometries inside Connes' spectral
  formalism and their relation to multimatrix models.  In this
  companion paper to [arXiv: 2007:10914, \href{https://doi.org/10.1007/s00023-021-01025-4}{\color{blue!80!black}\textit{Ann. Henri
      Poincar\'e}, 22: 3095–3148, 2021}] we propose a gauge theory setting based on
  noncommutative geometry, which---just as the traditional formulation
  in terms of almost $\!\!\!$-$\!\!\!$ commutative manifolds---has the ability to also
  accommodate a Higgs field. However, in contrast to
  `almost-commutative manifolds', the present framework employs only
  finite dimensional algebras which we call gauge matrix spectral triples.  In a path-integral quantization
  approach to the Spectral Action, this allows to state
  Yang-Mills--Higgs theory (on four-dimensional Euclidean fuzzy space)
  as an explicit random multimatrix model obtained here,
  whose matrix fields exactly mirror those of the Yang-Mills--Higgs theory
  on a smooth manifold.
\end{abstract}
\maketitle

  \fontsize{11.4}{14.0}\selectfont    

\section{Introduction} \label{sec:intro}

The approximation of smooth manifolds by \textit{finite}
\textit{geometries} (or geometries described by finite dimensional
algebras) has been treated in noncommutative geometry (NCG) some time
ago \cite{Landi:1999ey} and often experiences a regain of interest; in
\cite{DAndrea:2013rix,ConnesWvS}, for instance, these arise from
truncations of space to a finite resolution.  In an ideologically
similar vein but from a technically different viewpoint, this paper
addresses gauge theories derived from the Spectral Formalism of NCG,
using exclusively finite-dimensional algebras, also for the
description of the space(time).  This allows one to make precise sense
of path integrals over noncommutative geometries.  Although this
formulation is valid at the moment only for a small class of
geometries, the present method might shed light on the general problem
of quantization of NCG, already tackled using von Neumann's
information theoretic entropy in \cite{EntropySpectral} and
\cite{KhalkhaliQuantization}, by fermionic and bosonic-fermionic
second quantization, respectively.
\par

Traditionally, in the NCG parlance, the term `finite geometry' is
employed for an extension of the \textit{spacetime} or \textit{base
  manifold} (a spin geometry or equivalently
\cite{Reconstruction,RennieVarilly} a \textit{commutative spectral
  triple}) by what is known in physics as `inner space' and boils down
to a choice of a Lie group (or Lie algebra) in the principal bundle
approach to gauge theory.  In contrast, in the NCG framework via the
Spectral Action \cite{Chamseddine:1996zu}, this inner space---called
\textit{finite geometry} and denoted by $F$---is determined by a
choice of certain finite-dimensional algebra whose purpose is to
encode particle interactions; by doing so, NCG automatically rewards us with the
Higgs field.  Of course, the exploration of the right structure of the
inner space $F$ is also approached using other structures,
e.g. non-associative algebras
\cite{BaezHuerta,Furey:2010fm,Jordan,Todorov:2019hlc} for either the
Standard Model or unified theories, but in this paper we restrict
ourselves to (associative) NCG-structures.\par
 
Still in the traditional approach via \textit{almost-commutative
  geometries} $M\times F$ \cite{Stephan:2005uj,CCM,WvSbook}, the
finite geometry $F$ plays the role of discrete extra dimensions or
`points with structure' extending the (commutative) geometry $M$,
hence the name.  What is different in this paper is the replacement of
smooth spin geometries $M$ by a model of spacetime based on
finite-dimensional geometries (`finite spectral triples') known as
\textit{matrix geometry} or \textit{fuzzy geometry}
\cite{barrettmatrix}.  Already at the level of the classical action, these geometries
have some disposition to the quantum theory, as it is known from
well-studied `fuzzy spaces'
\cite{MadoreS2,DolanHuetOConnor,SteinackerFuzzy,SperlingSteinacker,Subjakova:2020haa,Steinacker:2020nva},
which are not always based on Connes' formalism\footnote{Also, other proposals
related to discretizations or truncations 
\cite{DAndrea:2013rix,GlaserStern,GlaserSternZwei,Bochniak:2020nxf}
are (closer to) spectral triples.}. This article lies in the
intersection and treats `fuzzy spaces' inside the Spectral Formalism.

At this point it is pertinent to clarify the different roles of the
sundry finite-dimensional algebras that will appear. Figure \ref{cubo}
might be useful to illustrate why matrix algebras that differ only in size are given different physical nature.  In this cube, pictorially similar to Okun's `cube for natural units' \cite{Okuncube,OkuncubeLandau}, classical Riemannian geometry sits at the origin $(0,0,0)$.
Several NCG-based theories of physical interest may have, nevertheless, the three more general coordinates $(\hbar,1/N,F)$
described now: 
\begin{figure}[h!]  \includegraphics[width=9.6cm]{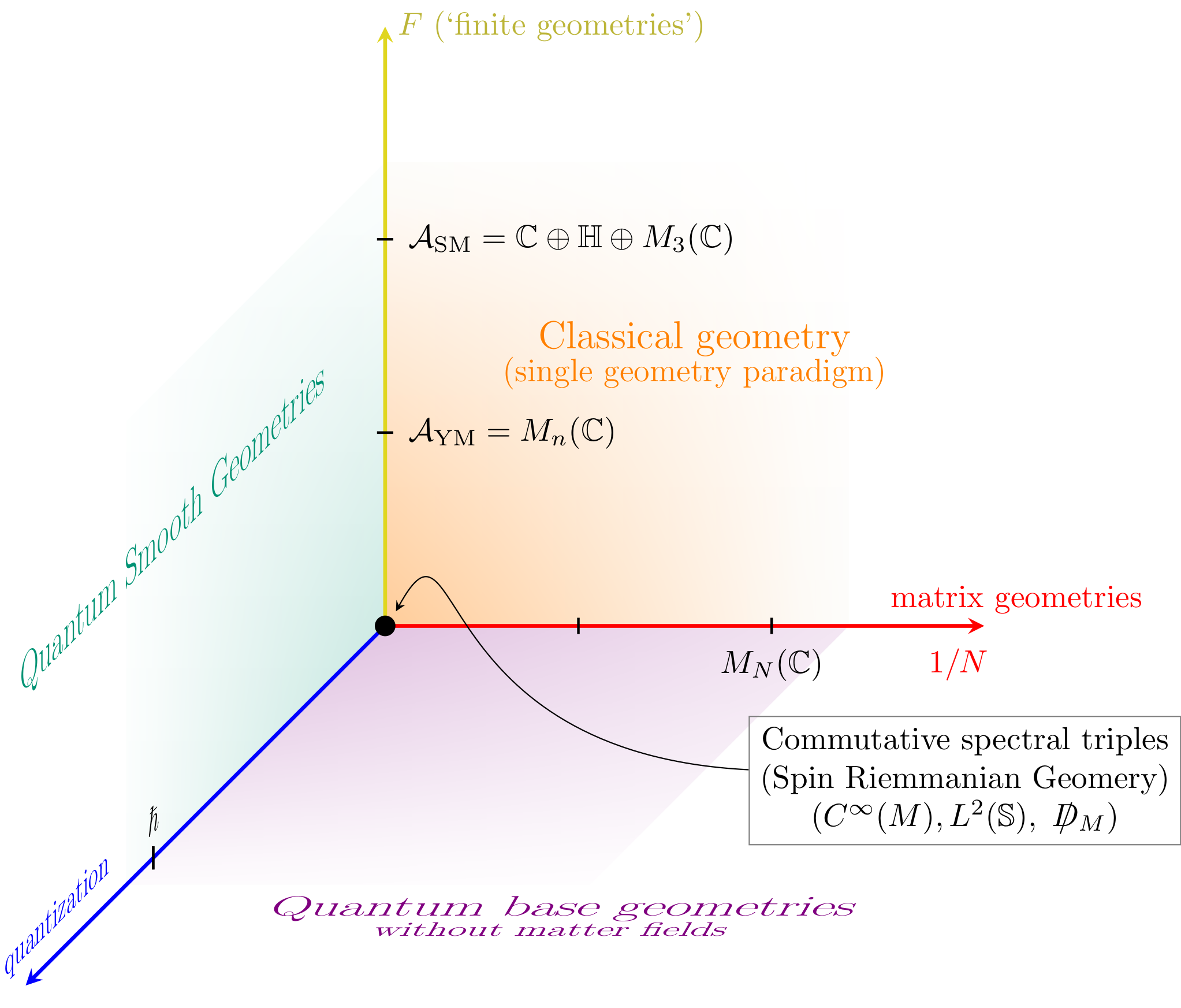}
\caption{Three axis representing independent theories (all inside NCG), 
starting from spin Riemannian geometry at the origin. 
Abbreviations and terminology:  YM=Yang-Mills; SM=Standard Model.
\label{cubo}}
\end{figure}
\begin{itemize} \setlength\itemsep{.4em}%
  \itemb The $F$-direction in Figure \ref{cubo} describes 
  (bosonic) matter fields. Mathematically the possible values for $F$
  correspond to a `finite geometry'. These were classified by Paschke-Sitarz \cite{PaschkeSitarz} and 
  diagrammatically by Krajewski
  \cite{KrajewskiDiagr}.  Particle physics
  models based on NCG and the Connes-Chamseddine spectral action
  \cite{CCM,BarrettSM,DAndrea:2013rix,DabrowskiSitarzDAndrea,BesnardUone,surveySpectral}
  `sit along the $F$-axis'. From those spectral triples $F$, only their algebra appears in Figure \ref{cubo}.

  \itemb A finite second coordinate, $1/N>0$, 
  means that the smooth base manifold that encodes
  space(time) has been replaced by a 
  `matrix geometry', which in the setting  \cite{barrettmatrix} is a spectral triple
  based on an algebra of matrices of size $N$ (and albeit finite-dimensional,
  escaping Krajewski's classification).

  \itemb The remaining coordinate denotes quantization when $\hbar\neq 0$.
  In the path integral formalism, the partition function is a weighted integral
  $Z=\int_{} \dif {\xi}\, \ee^{\ii S(\xi)/\hbar}$ over the space of
  certain class of geometries $\xi$, the aim being the quantization of
  space itself, having quantum gravity as motivation. Here $S$ is the classical action.

\end{itemize}
Accordingly, the planes orthogonal to the axis just described are:
\begin{itemize} \setlength\itemsep{.4em}%
  \itemb \textit{The plane $(\hbar,1/N,0)$ of base geometries}.
  On the marked plane orthogonal to $F$ lie
  `spacetimes' or\footnote{Here the name `base' for is taken from the principal 
  bundle $ G \hookrightarrow P \to M$ terminology, where $M$ is usually the spacetime manifold.} `base manifolds' and, when these are not flat, they
  can model gravity degrees of freedom. If $F=0$, no gauge fields live on such space.
  
     \itemb \textit{The plane $(\hbar , 0, F)=\lim_{N\to \infty} (\hbar,1/N,F) $}.
  On the plane orthogonal to the `matrix geometry' 
  axis,
  one has the quantum, smooth geometries (meaning, their algebra is or contains 
  a $C^\infty(M)$ as factor). The long-term aim is to get to the
  `quantum smooth geometry plane' as matrix algebras become
  large-dimensional, which is something that, at least
  for the sphere, is based on sound statements \cite{Rieffel:2001qw,Rieffel:2007hv,Rieffel:2015aya, Rieffel:2017hbl}
  in terms of Gromov-Hausdorff convergence. Additional to such large-$N$ one might require to
  adjust the couplings to criticality
  \cite{BarrettGlaser,Glaser:2016epw,KhalkhaliPhase}. This can also be
  addressed using doubly scaling limits together with the Functional Renormalization Group to find
  candidates for phase transition; for models still without matter,
  see \cite{FRGEmultimatrix}.
  
  \itemb \textit{The plane $(0,1/N,F)=\lim_{\hbar\to 0} (\hbar , 1/N,F)$ of classical geometries}.
  By `classical geometry' we mean a single
  geometrical object (e.g.  a Lorentzian or Riemannian manifold, a 
  $\mathrm{SU}(n)$-principal bundle with connection, etc.), which can
  be determined by, say, the least-action principle (Einstein
  Equations, $\mtr{SU}(n)$-Yang-Mills Equations, etc.). This is in
  contrast to the quantization of space, which implies a
  multi-geometry paradigm, at least in the path integral approach.
\end{itemize}

The program started here is not as ambitious as to yield physically
meaningful results in this very article, but it has the initiative to
apply three small steps---one in each of the independent directions
away from classical Riemannian geometry---and presents a model in
which the three aforementioned features coexist. This paves the way
for NCG-models of quantum gravity coupled to the rest of the
fundamental interactions (it is convenient to consider the theory as a
whole, due to the mutual feedback between matter and gravity sectors
in the renormalization group flow; cf.  \cite{Dona:2013qba} for an
asymptotic safety picture). For this purpose we need the next
simplifications, illustrated in Figure \ref{fig:cubodos}:

\begin{itemize} \setlength\itemsep{.4em} \itemb Our choice for the
  finite geometry $F$ is based on the algebra $\A_F=\Mn$ ($n\geq 2$).
  This is the first input, aiming at a $\mathrm{SU}(n)$ Yang-Mills
  theory.

  \itemb Instead of the function algebra on a manifold, we take a
  simple matrix algebra $M_N(\C)$. This is an input too. (Also $N$ is
  large and $n$ need not be.)

  \itemb We use random geometries instead of honest quantum
  geometries; this corresponds with a Wick rotation from
  $ \ee^{\ii S(\xi)/ \hbar } $, in the partition function, towards the
  Boltzmann factor $ \ee^{- S(\xi)/\hbar}$.  This setting is often referred to as
  \textit{random noncommutative geometry}
  \cite{Glaser:2016epw,BarrettDruceGlaser}.
\end{itemize}
\begin{figure}[h]
\includegraphics[width=9cm]{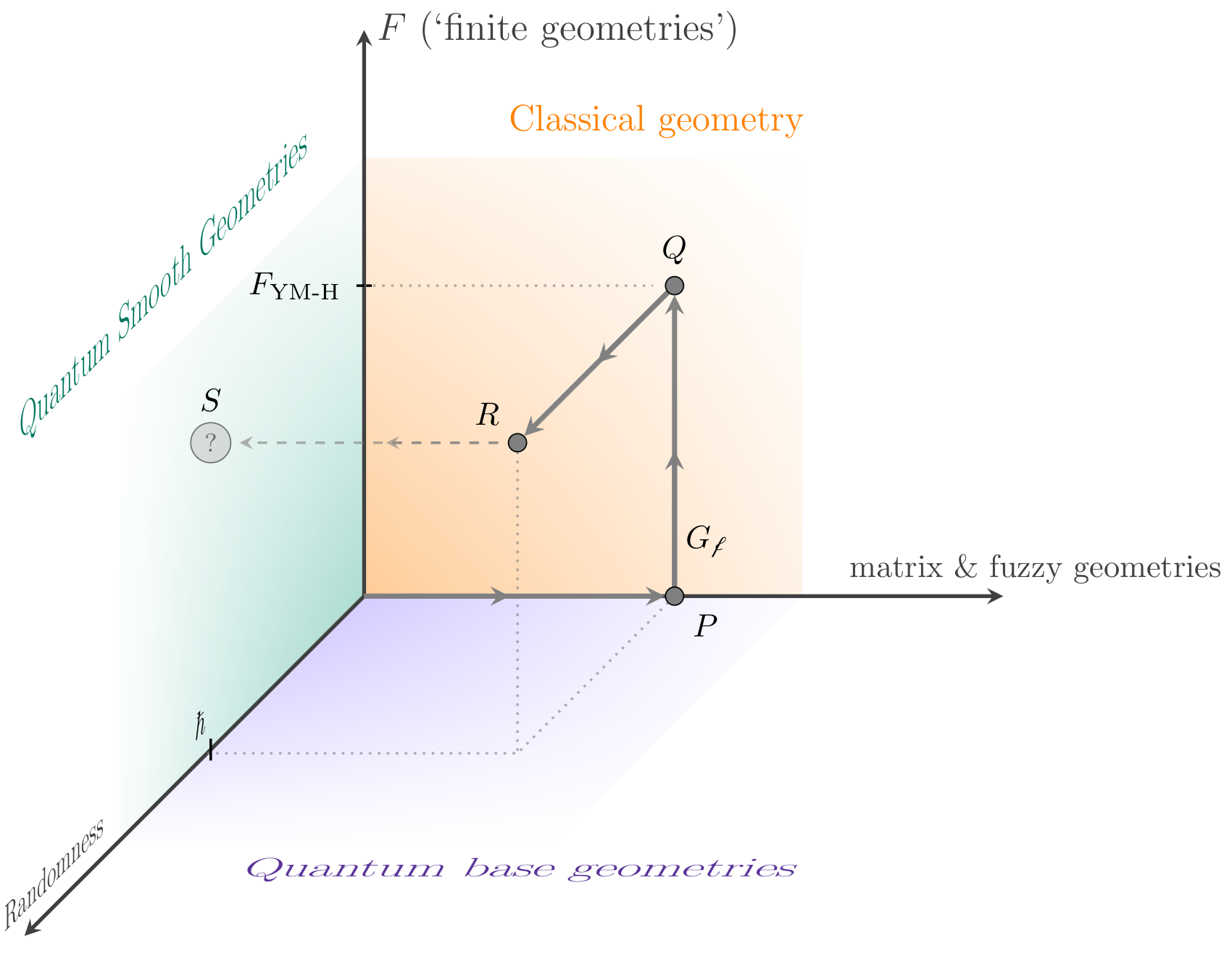}
\caption{
Depicting the organization of this article,
following the path $PQR$. Here, $F_\YMH={(M_n(\mathbb C), M_n(\mathbb C),D_F)}$ 
corresponds to the spectral triple for the Yang-Mills--Higgs theory
and $G_\fay$ is a fuzzy 4-dimensional geometry. As outlook (dashed),  
to reach a smooth geometry
at the point $S$ one needs a sensible limit 
(e.g. large-$N$ and possibly tuning some parameters to criticality) in order to achieve phase transition\label{fig:cubodos}
}
\end{figure}

Random NCG was introduced in \cite{BarrettGlaser}.  While aiming at
numerical simulations for the Dirac operators, Barrett-Glaser stated
the low-dimensional geometries as a random matrix model.  The Spectral
Action of these theories was later systematically computed for general
dimensions and signatures in \cite{SAfuzzy}.  Also, in the first part
of this companion paper, the Functional Renormalization Group to
multimatrix models \cite{FRGEmultimatrix} inspired by random
noncommutative geometry was addressed for some two-dimensional models
obtained in \cite{SAfuzzy}. Solution of the matrix-models
corresponding to one-dimensional geometries was addressed in
\cite{KhalkhaliTR}, using Topological Recursion \cite{EynardOrantin}
(due to the presence of multitraces, in its blobbed
\cite{BorotBlobbed} version).

The organization of the article is as follows. Next section introduces
fuzzy geometries as spectral triples and gives Barrett's
characterization of their Dirac operators in terms of finite
matrices. Section \ref{sec:MatrixSpin} interprets these as variables
of a `matrix spin geometry' for the $(0,4)$-signature.  Section
\ref{sec:GeneralFinACGeom} introduces the main object of this article,
\textit{gauge matrix spectral triples}, for which the spectral
action is identified with Yang-Mills theory, if the piece $D_F$ of
Dirac operator along the `inner space spectral triple'\footnote{This
  is usually referred to as `finite spectral triple' but in this paper
  all spectral triples are finite dimensional.} vanishes, and with
Yang-Mills--Higgs theory, if this is non-zero, $D_F\neq 0$ (see
Sec. \ref{sec:YMH}).  Our cutoff function $f$ appearing in the
Spectral Action $\Tr_\H f(D)$ is a polynomial $f$ (instead of a bump
function\footnote{This is not the first time that the Connes-Chamseddine regulating
function $f$ does not appear and instead a polynomial is used, e.g. see the approach by
\cite{MvS} in the spin network context.}).  In Section
\ref{sec:smoothlim} we make the parallel of the result with ordinary
gauge theory on smooth manifolds.  Finally, Section \ref{sec:Conclusions}
gives the conclusion and  Section \ref{sec:Outlook} the outlook, while also stating the explicit
Yang-Mills--Higgs matrix model for further study.
\par

This article is self-contained, but some familiarity with spectral
triples helps. Favoring a particle physics viewpoint, we kept the
terminology and notation compatible with \cite{WvSbook}.
\tableofcontents

\section{Spectral Triples and Fuzzy Geometries} \label{sec:STandFuzzy} 

Let us start with Barrett's definition of fuzzy geometries that makes
them fit into Connes' spectral formalism.
 
\begin{definition} \label{def:fuzzy} A \textit{fuzzy geometry} is
  determined by
  \begin{itemize} \setlength\itemsep{.4em} \itemb a \textit{signature}
    $(p,q)\in \Z_{\geq 0}$, or equivalently, by
    \[\eta =
      \diag(\underbrace{+,\ldots,+}_{p},\underbrace{-,\ldots,-}_q)=\diag(+_p,-_q)\]
    \itemb three signs $\epsilon, \epsilon',\epsilon''\in \{-1,+1\}$
    fixed through $s$ by the following table: \vspace{1pt}

  \centering
  \vspace{5pt}
\begin{tabular}{ccccccccc}
  \topline
  \headcol $s\equiv q-p \,\,\mtr{ mod }\, 8$ & 0 & 1 & 2 & 3 & 4& 5 &6 &7  \\
  \midline
  $\epsilon$ & $+$ & $+$ & $-$ & $-$ &$-$&$-$& $+ $ &$+ $ \\
 $\epsilon'$ & $+$ & $-$ & $+$ & $+$ &$+$&$-$&$+$&$+$ \\
  $\epsilon''$ & $+$ & + & $-$ & + &$+$& +&$-$& +\\
  \bottomlinec
\end{tabular}
\flushleft
\end{itemize}
\begin{itemize} \setlength\itemsep{.4em} \itemb a matrix algebra
  $\A\fuz=\MN$ \itemb a Clifford $\mathcal{C}\ell(p,q)$-module $V$ or
  \textit{spinor space} \itemb a \textit{chirality}
  $\gamma\fuz=\gamma\otimes 1_\A:\H\fuz \to \H\fuz$ for the vector
  space $\H\fuz = V\otimes M_N({\C})$ with inner product
  \[\qquad \quad\langleb v\otimes T, w\otimes W \rangleb = (v,w) \TrN(T^*W)\, \text{}\,  \]
  for all $T,W\in \MN$ and $v,w\in V$. 
  To wit $\gamma:V\to V$ is self-adjoint with respect to the Hermitian
  form $(v,w)= \sum_a \bar v_a w_a$ on $V \cong \C^{k}$ and satisfying
  $\gamma^2=1 $. This $k$ is so chosen as to make $V$ irreducible for
  even $s$. Only the $\pm1$-eigenspaces of $V$ with the grading
  $\gamma$ are supposed to be irreducible, if $s$ is odd

  \itemb a left-$\A\fuz$ \textit{representation} on $\H\fuz$,
  $\varrho(a)(v\otimes W) = v\otimes (a W)$, for $a\in \A\fuz$ and
  $W\in \MN$.  The representation $\varrho$ is often implicit
 
  \itemb an anti-linear isometry, called \textit{real structure},
  $J\fuz:=C\otimes *:\H\fuz \to \H\fuz$ given in terms of the
  involution $*$ (in physics represented by $\dagger$) on the matrix
  algebra and $C:V\to V$ an anti-linear operator satisfying, for each
  gamma matrix,
 \begin{align}\label{C}
 C^2= \epsilon   
 \mbox{ and }\gamma^\mu C = 
 \epsilon ' C \gamma^\mu  
 \end{align}

\itemb  a self-adjoint operator $D$ on $\H$
satisfying the \textit{order-one condition}
\begin{equation}
 \big[\hspace{.51pt}[ D\fuz ,\varrho(a)]\,, J\fuz\varrho(b)J\fuz\inv\big]=0 \qquad \mbox{for all }a,b\in \A
\label{orderone}
\end{equation}

\itemb  the condition\footnote{This condition
does not appear in list given by Barrett and in fact 
follows from the construction of the explicit $\gamma$ matrices,
so it is tautological but useful to emphasize, as it also appears in 
the smooth case \cite{WvSbook}. Barrett
also allows algebras $\A$ over $\re,\C$ and $\mathbb H$; 
for quaternion coefficients, $M_{N/2}(\mathbb H)\subset  M_{N}({\C})$.}  $D\gamma\fuz  = -\gamma\fuz D  $ for even $s$. Moreover, the three 
signs above impose: \vspace{-.4cm}\begin{subequations}%
   \begin{align}%
    J^2\fuz& =\epsilon \, , \label{redundante}\\
 J\fuz D\fuz & =\epsilon' D\fuz J\fuz \,,  \\
 J\fuz\gamma\fuz & =\epsilon '' \gamma\fuz J\fuz\,.
 \end{align}%
\label{signos}%
\end{subequations}%
 \vspace{-.4cm}
\end{itemize}%
Notice that, in this setting, the square of $J\fuz$ is obtained from
$C$ as specified above, but we added the redundant \eeqref{redundante},
as this equation appears so for general real, even spectral triples.
For $s$ odd, $\gamma\fuz$ can be trivial $\gamma\fuz=1_{\H}$.  The
number $d:=p+q$ is the \textit{dimension} and $s:=q-p$ (mod $8$) is
the \textit{KO-dimension}.
\end{definition}%
\begin{remark} \label{rmk:commrel} It will be useful later to stress
  that the `commutant property' (cf. for instance
  \cite[eq. 4.3.1]{WvSbook})
\begin{align}
\label{LRaxiom}
[a,J b^* J\inv ] =0\,,\qquad \text{for all } a,b \in\A \,,
\end{align}
which is typically an axiom for spectral triples, is not assumed in our
setting. However, one can show that it is a consequence of those in
Definition \ref{def:fuzzy}. The axiom states that the right
$\A$-action $\psi b:=b^{\mtr o} \psi = J b^* J\inv \psi $, for
$b \in \A, \psi \in \H$, commutes with the left $\A$-action 
$\varrho$ for each $a,b\in \A $. Since $J=C\otimes *$, and the algebra
acts trivially on $V$,
\begin{align} \label{RLcommute}
 a b^{\mtr o} (v\otimes m ) = 
 a J  ( v \otimes b^* m^* ) &=   
  v \otimes (a  m b  ) \\
 & = b^{\mtr o}  a (v\otimes m)\,,  \quad v \in V, m\in \MN  \,. \nonumber
\end{align}
\end{remark}

The focus of this paper is dimension four, but we still proceed in
general dimension.
We impose on the 
gamma matrices $\gamma^\mu$ the following conditions:
\begin{subequations} \label{emusquare}
\begin{align}
(\gamma^\mu)^2&=+1_V, &&  \text{and } \gamma^\mu  \mbox{ Hermitian for } & & \mu=1,\ldots, p,   \\ 
(\gamma^\mu)^2&=-1_V, && \text{and } \gamma^\mu \mbox{ anti-Hermitian for } & & \mu=p+1,\ldots,p+ q.
\end{align}
\end{subequations}%
Since it will be convenient to treat several signatures simultaneously,
we let 
$(\gamma^\mu)^2=: e_\mu 1_V$ for each $\mu=1,\ldots,d$.
According to Eqs. \eqref{emusquare}, one thus obtains 
the unitarity of all gamma-matrices:
\[
(\gamu v,\gamu w)= 
( (\gamu)^* \gamu  v,  w)
= (e_\mu \gamu \gamu  v, w)
=(e_\mu)^2 ( v, w)=( v, w) \qquad 
\]
without implicit sum, and for each $v,w \in V$. Let these matrices generate
$\Omega:= \langleb \gamma^1,\ldots,\gamma^d \rangleb_\re$ as algebra,
for which one obtains a splitting $ \Omega = \Omega^+\oplus \Omega ^-$
where $\Omega^\pm$ is contains products of even/odd number of
gamma-matrices.  According to \cite[Eq. 64]{barrettmatrix}, the Dirac
operator $D\fuz$ solves the axioms of an even-dimensional fuzzy
geometry whenever it has the next form:
\begin{align} \label{DiracCharact}\raisetag{8pt}
 D\fuz (v\otimes T) & =
\sum_{I } \gamma^I v  \otimes \{ K_I , T \}_{e_I}\text{ and } e_I\in \{+1,-1\}\,,\\
\{A,B\}_{\pm}& := AB \pm BA  \,,\nonumber
\end{align} 
where $T\in \MN$ and the sum is over increasingly ordered
multi-indices $I=(\mu_1,\ldots,\mu_{2r-1})$ of odd length.  With such
multi-indices $I$ the following product
$\gamma^I:=\gamma^{\mu_1}\cdots \gamma^{\gamma_{2r-1}} \in \Omega^-$
is associated (the sum terminates after finitely many terms, since
gamma-matrices square to a sign times $1_V$). Moreover, still as part
of the characterization of $D\fuz$, $e_I$ denotes a sign chosen
according to the following rules:
\begin{itemize} \setlength\itemsep{.4em} \itemb if $\gamma^I$ is
  anti-Hermitian (so $e_I=-1$), then $\{ K_I , T \}_{e_I}=[L_I,T] $,
  i.e. $ \{ K_I , \balita \}_{e_I}$ is a commutator of the
  anti-Hermitian matrix $K_I$ (denoted by $L_I$); and \itemb if
  $\gamma^I$ is Hermitian, so must be $K_I$, which will be denoted by
  $H_I$. Then $e_I=+1$, and $\{ K_I , T \}_{e_I} =\{H_I,T\}$, so
  $\{ K_I , \balita \}_{e_I}$ is an anti-commutator with a Hermtian
  matrix $H_I$.

\end{itemize}

\begin{example} \label{ex:DiracSignatures}
Some Dirac operators of fuzzy $d$-dimensional geometries, $d=2,3,4$
in several `types' (or signatures) $(p,q)$.

\begin{itemize} \setlength\itemsep{.4em}
 \itemb  \textit{Type} (0,2). Then $s=d=2$, so $\epsilon'=1$.
 The gamma matrices are anti-Hermitian
 and satisfy $(\gamma^i)^2=-1$. The Dirac operator is
 \[ D\hp{0,2}\fuz=\gamma^1 \otimes [L_1,\balita] +\gamma^2 \otimes [L_2,\balita]
\] 
\itemb \textit{Type} (0,3), $s=3$.  In this signature, the gamma
matrices can be replaced for the quaternion units $\iay,\jay$ and
$\kay$ to express the $(0,3)$-geometry Dirac operator
as\footnote{\label{foot:simplification}This formula differs from the
  most general $(0,3)$-geometry Dirac operator
  \cite[Eq. 73]{barrettmatrix} spanned by eight gamma matrices, since
  ours corresponds to a simplification (also addressed in \S V. A of
  op. cit.) byproduct of $V$ being irreducible and the product of all
  gamma matrices being a scalar multiple of the identity.}
\[D\hp{0,3}\fuz=\{H,\balita\} + \iay [L_1,\balita]+ \jay
  [L_2,\balita]+\kay [L_3,\balita] \] \itemb \textit{Type} (0,4),
$s=4$, \textit{Riemannian}.  Since the triple product of anti-Hermitian gamma
matrices is self-adjoint,
$(\gamma^\alpha\gamma^\mu\gamma^\nu)^*=(-)^3\gamma^\nu\gamma^\mu\gamma^\alpha=\gamma^\alpha\gamma^\mu\gamma^\nu$,
so are the operator-coefficients, which have then the form
$ \{H_{\alpha\mu\nu},\balita\}$ for
$(H_{\alpha\mu\nu})^*=H_{\alpha\mu\nu}$:
        \begin{align*}
     \qquad D\hp{0,4}\fuz& =  \sum_\alpha\gamma^\alpha \otimes [L_\alpha,\balita]+ \sum_{\kappa<\lambda <\mu} 
\underbrace{\gamma^\kappa  \gamma^\lambda \gamma^\mu }_{\gamma^{\hat \rho}}\otimes  \{\underbrace{H_{\kappa\lambda \mu}}_{H_{\hat\rho}},\balita\}
     \\ 
     &= \sum_\rho\gamma^\rho \otimes [L_\rho,\balita]+  
     \gamma^{\hat \rho} \otimes  \{H_{\hat \rho},\balita \}
     \quad \text{ ($\{\rho,\kappa,\lambda, \mu\}=\{0,1,2,3\}$) }
     \end{align*}
     where $\gamma^{\hat\rho}$ means the product of gamma matrices
     with indices different from $\rho$, multiplied in ascending order; see the restriction 
     in the sum in the expression for $D\hp{0,4}\fuz$.
     \itemb  \textit{Type}  (1,3), $s=2$, \textit{Lorentzian}.  Let $\gamma^0$
       be the time-like gamma matrix, i.e. the only one squaring to $+1$. Then
  \begin{align} \nonumber
     D\hp{1,3}\fuz & = \gamma^0 \otimes \acomm{H_0}+
     \sum_i \gamma^i  \otimes   \comm{L_i} \\[-12pt]
     &   + \sum_{i<j}\gamma^0 \gamma^i\gamma^j \otimes \comm{L_{ij}}
       +  \overbrace{\gamma^1  \gamma^2 \gamma^3}^{\gamma^{\hat 0}}\otimes  \{H_{{\hat 0}} ,\balita\} 
        \end{align}
\end{itemize}
\end{example}
In the sequel we use $K_I$ generically for either $H_I$ or $L_I$,
whose adjointness-type is then specified by the signature and by
$I$. We also define the sign $e_{I}$ by $K_I ^*:= e_I K_I$, or
equivalently by $(\gamma^I)^*= e_I \gamma^I$, for a multi-index
$I$. In four dimensions, one has for triple indices $I=\hat \mu$
\cite[App. A]{SAfuzzy}
\begin{align}
  e_{\hat \mu}
= e_\mu (-1)^{q+1}\hspace{1cm} 1\leq\mu\leq d=p+q=4,\mbox{ for signature } (p,q)\,. \label{mhm}
\end{align}
In summary, a fuzzy geometry of signature ($p,q$) has following
objects:
\begin{itemize}\setlength\itemsep{.4em}
  \itemb $\A\fuz=\MN$
  \itemb $\H\fuz=V \otimes \MN$, Hilbert-Schmidt
  inner product on $\MN$
  \itemb a representation of $\A\fuz$ on
  $\H\fuz$, $\varrho(a) (v\otimes T)= v\otimes a T$
  \itemb $D\fuz $
  given by \eeqref{DiracCharact}
  \itemb $J\fuz=C\otimes *$ with $C$
  anti-linear satisfying \eeqref{C}
  \itemb  $\gamma\fuz=\gamma \otimes 1_{\MN}$, with $\gamma$ constructed from
  all $\gamma$-matrices; see \eeqref{gammadef} for $d=4$
 \end{itemize}
Although next equation is well-known, we recall it due of its recurrent usefulness later. 
In any dimension and signature, it holds:
\begin{align}\label{tr_4gammas} 
\TrV(\gamma^\mu\gamma^\nu \gamma^\alpha \gamma^\rho) &= \dim V \cdot  \bigg(
\raisebox{-.40\height}{\includegraphics[height=1.2cm]{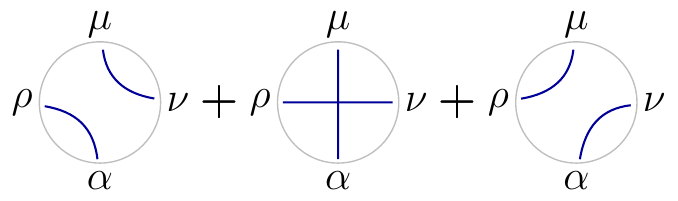}  } \!\! \bigg) \\
& = \dim V \cdot  ( \eta^{\mu\nu} \eta^{\alpha \rho}
-\eta^{\mu \alpha } \eta^{ \nu\rho }
+\eta^{\mu \rho} \eta^{\nu \alpha}
) \nonumber 
\end{align}
Each inscribed segment in the \textit{chord diagrams} denotes an
index-pairing between two indices labeling their ends, say $\lambda$
and $\theta$, which leads to $\eta^{\lambda \theta}$; all the pairings
of each diagram are then multiplied bearing a total sign corresponding
to $(-1)$ to the number of simple chord crossings.  This picture is
helpful to compute traces of more gamma-matrices, but is not essential
here; see \cite{SAfuzzy} to see how the spectral action for fuzzy
geometries was computed by associating with these chord diagrams
noncommutative polynomials in the different matrix blocks $K_I$
composing the Dirac operator. Incidentally, notice that so far this
chord diagram expansion is classical, unlike that treated by Yeats in
\cite[\S 9]{Yeats}, which appears in the context of Dyson-Schwinger
equations.  \allowdisplaybreaks[2]
\section{Towards a `Matrix Spin Geometry'}\label{sec:MatrixSpin}

We restrict the discussion from now on to dimension four, leaving the
geometry type (KO-dimension) unspecified.  Next, we elaborate on the
similarity of the fuzzy Dirac operator and the spin-connection part
spanned by multi-indices, which has been sketched in \cite[Sec. V \S
A]{barrettmatrix} for $d=4$.  The identification works only in
dimensions four and, if `unreduced' (cf. Footnote
\ref{foot:simplification} above) also three. For higher dimensions, quintuple
products appear; for lower ones, triple products are absent. Although
it would be interesting to address each dimensionality separately,
since the physically most interesting case is dimension $4$, we stick to
it.  \vspace{.2cm}

\noindent 
\begin{remark}
Since some 
    generality might be useful for the future, or
    elsewhere (e.g. in a pure Clifford algebra context), even though
    we identify the geometric meaning only for the objects in
    Riemannian signature, we prove most results in general
    signature.
\end{remark}

For a Riemannian spin manifold $M$, recall the local expression (on an
open $U\subset M$) of the canonical Dirac operator on the spinor
bundle $\mathbb S\to M$ for each section $\psi$ there,
\begin{subequations}
\begin{align} (D_M \psi) (x) &=  \ii\Gamma^j(x) \nabla_j^S \psi(x)\,, \text{ for $x\in U$ and $\psi \in 
\Gamma^\infty (U,\mathbb S)$,} \qquad \\
  \nabla_i^S& = \partial_i + \omega_i\,.
\end{align}
\end{subequations}%
The coefficients
$ \omega_i= \frac{1}{2} \omega_i^{\mu\nu} \gamma_{\mu\nu}$ of the spin
connection $\nabla^S$ (the lift of Levi-Civita connection) are here
expressed with respect to a base
$\gamma_{\mu\nu}= \frac14 [\gamma_\mu, \gamma_\nu]$ that satisfies the
$\mathfrak o(4)$ Lie algebra in the spin representation (see e.g. \cite[\S
11.4]{ConnesMarcolli}). The gamma matrices with Greek indices (or
`flat') $\gamma^\mu$ relate to the above
$\Gamma^i(x) = e^i_\mu \gamma^\mu $ by means of \textit{tetrads}
$e^i_\mu (x) $. The coefficients $e_\mu^i\in C^\infty (U)$, by
definition, make of the set of fields
$(E_\mu)_{\mu=0,1,2,3}= (e_\mu^i \cdot \partial_i)_{\mu=0,1,2,3}$ an
orthonormal basis of $\mathfrak{X}(U)$ with respect to the metric $g$
of $M$, which is to say $g(E_\mu,E_\nu)= \eta_{\mu\nu}$.  Thus
$\{\Gamma^i(x), \Gamma^j (x) \}=2g^{ij}(x)=2(g\inv)_{ij}(x) $ for
$x\in U$, but $\{\gamma^\mu,\gamma^\nu\}=2\eta^{\mu\nu}$.  In contrast
to the commutation relations that the elements of the coordinate base
$\partial_i=\partial/ \partial x^i$ satisfy, one generally has
$[E_\mu,E_\nu]\neq 0$ for the non-coordinate base $E_0,\ldots, E_3$,
also sometimes called \textit{non-holonomic} \cite[\S
4]{TorresDiffManifolds}.  Notice that in the fuzzy setting only Greek
indices 
appear. \begin{quote} \textit{This, together with the fact that rather
    $\eta^{\mu\nu}$ instead of $g^{ij}$ appears in the Clifford
    algebra, should \underline{not} be interpreted at this stage as
    flatness. Instead, for fuzzy geometries the equivalent of a metric
    is encoded in the the signature $\eta=\diag (e_0,\ldots, e_3)$
     \underline{and} in the matrices parametrizing the Dirac
    operator. }\end{quote}
In Riemannian signature, we rewrite\footnote{ The restrictions
  $0 < \sigma\leq 3$ and $ 1<\nu\leq 3$ account for the appearance in
  \eeqref{Fuzzyconnection} of exactly three gamma matrices whose
  indices are increasingly ordered, as in the characterization
  \eqref{DiracCharact} of fuzzy Dirac operators.  To match the
  canonical Dirac operator on a spinor bundle, one could redefine $H$,
  fully anti-symmetrize, and compare expressions.}
(cf. Ex. \ref{ex:DiracSignatures})
\begin{subequations}
\begin{align}
D\fuz & = \sum_\mu (\gamma^\mu\otimes 1_N ) (\nabla^{S}\fuz)_\mu \,, \\
 (\nabla^{S}\fuz)_\mu & =  1_V \otimes  [ L_\mu, \balita ] 
+ \sum_{\substack{0 < \sigma\leq 3 \\ (\mu < \sigma)}} \, \sum_{\substack{1 < \nu \leq 3\\ ( \sigma <\nu)}} \gamma^ \sigma \gamma^\nu \otimes 
 \{ H_{\mu \sigma\nu}, \balita \}_{e_{\mu \sigma\nu} } \,. \label{Fuzzyconnection}
\end{align}%
\label{Dirac4d_explicit}%
\end{subequations}%
Simultaneously (up to the trivial factor $1_V$), we identify the
commutators $ [ L_\mu,\balita ] $ with
$\ii E _\mu=\ii e_\mu^j \partial_j $ and the coefficients of the
triple gamma products
$ \{H_{\mu \sigma\nu}, \balita \}_{e_{\mu \sigma\nu} }$ with the full
anti-symmetrization
$\frac{\ii}{4} \omega_{[\mu| ik }^{\phantom{.}} e^i_{|\sigma}
e^k_{\nu]} $ of the spin connection coefficients in the three Greek
indices.  The triple products of gamma-matrices present in the Dirac
operator \eqref{Dirac4d_explicit} are the analogue of those in the
spin connection appearing in
$D_M= \ii \gamma^\mu ( E_\mu + e_\mu^i \omega_i)$, here in the `flat'
(non-holonomic or non-coordinate) basis $E_0,\ldots,E_3$. Altogether,
$\nabla^{S}\fuz$ can be understood as the matrix spin connection.

We let $\Delta_4=\{0,1,2,3\}$ and denote by
$ \delta_{\mu\nu\alpha\sigma}$ the fully symmetric symbol with indices
in $\Delta_4$, which is non-vanishing (and then equal to 1) if and
only if the four indices are all different; equivalently,
$\delta_{\mu\nu\alpha\sigma} = |\epsilon_{\mu\nu\alpha\sigma}|$, in
terms of the (flat) Levi-Civita symbol $\epsilon$.
\\

\noindent
\textsc{Remark on notation.} 
Specially when dealing with fuzzy geometries, we sometimes do not use Einstein's summation (traditional in differential geometry). 
We avoid raising and lowering indices as well, e.g. gamma matrices are presented only
with upper indices.   We set $k=(k_\mu)_{\mu \in \Delta_4}, K=(K_\mu)_{\mu \in \Delta_4}, x=(x_\mu)_{\mu\in \Delta_4}$, et cetera.

\begin{lemma}\label{thm:gamma13}
For any $\mu,\nu\in \Delta_4=\{0,1,2,3\}$ the following relations are satisfied
for any signature $\eta=\diag(e_0,e_1,e_2,e_3)$ in four dimensions:  
\begin{subequations} \label{extwolemmas}
 \begin{align}\label{gammas13}
\gamma^\mu \gamma^{\hat \nu}
&= 
(-1)^{\mu}\Big ( 
\delta^{\mu}_\nu \gamma^0 \gamma^1\gamma^2\gamma^3
+\mtr{sgn}(\nu-\mu)\sum_{\alpha < \sigma} \delta_{\mu\nu\alpha\sigma} e_\mu \gamma^\alpha \gamma^\sigma
\Big)
\,, \\[-4pt]
\gamma^{\hat \mu}\gamma^\mu \label{gammas13b}
&= - \gamma^\mu \gamma^{\hat \mu}
\,, \\[5pt]
 \gamma^{\hat \nu}\gamma^\mu
&= +\gamma^\mu \gamma^{\hat \nu} \qquad (\nu\neq\mu)
\,,  \label{gammas13c} \\
\label{eq:gammastriples}
 \gamma^{\hat \mu }\gamma^{\hat \nu } &
 = (-1)^{1+|\mu-\nu|}\sum_{\rho,\lambda}\frac{1}{2}
 \delta_{\mu\nu\lambda\rho}e_\lambda e_\rho \gamma^\mu\gamma^\nu -1_V 
\delta^\mu_{\nu} \cdot e_\mu\cdot  \det (\eta) \,.
\end{align}  
\end{subequations}
\end{lemma}

This lemma is proven in Appendix \ref{app:prooflemma}. 
Notice that in Eqs. \eqref{extwolemmas} the repeated indices
$\mu,\nu$ in the RHS are not summed (therefore the index-symmetry of
$\delta_{\mu\nu\lambda\rho} $ with the antisymmetry of
$\gamma^\mu \gamma^\nu$ does annihilate that term).

We now need the explicit form of the chirality 
$ \gamma\fuz  =\gamma\otimes 1_{\MN}$, 
given by
\begin{align}\label{gammadef}
 \gamma  =  (-\ii)^{\frac{1}{2}(q-p)(q-p+1)} \gamma^0\gamma^1\gamma^2\gamma^3 
 =:  \sigma(\eta)  \gamma^0\gamma^1\gamma^2\gamma^3 \,.
\end{align}
This factor $ \sigma(\eta)$ in $\gamma$ in front of the matrices is $-1,+\ii,+1,-\ii$,
for the signatures $(p,q)= (0,4), (1,3), (2,2),(3,1)$, respectively,
corresponding to KO-dimensions $s=4,2,0,6$.

\begin{lemma}\label{thm:Lichnerowicz}
The square of the Dirac operator of a fuzzy geometry $G\fuz$ of signature 
$\eta=\diag(e_0,\ldots,e_3)$ is
\begin{align}
D^2\fuz&=
  \sum_{\mu,\nu}1_V \otimes \eta^{\mu\nu }k_\mu \circ k_\nu  + \frac{1}{2}\gamma^\mu\gamma^\nu \otimes [k_\mu, k_\nu]_\circ
\nonumber
-
\sum_\mu   \det(\eta)  e_\mu 1_V\otimes x_\mu\circ x_\mu
\\ 
&
\label{Lichnerowicz}
+ \sum_{\mu < \nu } 
t_{\mu\nu} \gamma^{\mu} \gamma^\nu \otimes [x_\mu,x_\nu]_{\circ}
+\frac12\sum_{\mu,\nu,\sigma,\alpha } s_{\mu\nu\alpha\sigma} \cdot \gamma^\alpha \gamma^\sigma \otimes \{ x_\nu,k_\mu\}_{\circ}
\\ \nonumber
&+\frac1{\sigma(\eta)}  \sum_{\mu} {(-1)^\mu}  \gamma \otimes [x_\mu,k_\mu]_{\circ}\,,
\end{align}
with the `commutator' $[f,g]_\circ$ given by $f\circ g - g\circ f$ in
terms of the composition $\circ$ of the following operators (which are themselves commutators or
anti-commutators)
\begin{align}\label{minusculas}
k_\mu:= \{ K_\mu,\balita \}_{e_\mu}\, \quad \mtr{and}  \quad
x_\mu:= \{ K_{\hat \mu},\balita \}_{e_{\hat \mu}}\,.
\end{align}
We defined also the (whenever non-vanishing) signs 
\begin{align}
  s_{\mu\nu\alpha\sigma}&:= e_\mu (-1)^{\mu} \cdot \sgn(\nu-\mu )\cdot  \sgn(\sigma-\alpha) \cdot \delta_{\mu\nu\alpha\sigma}\in \{-1,0,+1\} \,,\\
  t_{\mu\nu}&:=  \sum_{\lambda  < \rho}(-1)^{1+|\mu-\nu|}\delta_{\mu\nu\lambda\rho}e_\lambda e_\rho \in \{-1,0,+1\}\,.
\end{align}
\end{lemma}

\begin{proof}
One straightforwardly finds $
 D^2 \fuz= 
(\mathfrak{a}+\mathfrak{b}+\mathfrak{c}+\mathfrak{d}+\mathfrak{e})(k,x)$ 
with  
\begin{subequations}
\begin{align}
 \mathfrak{a}(k,x)&= \sum_{\mu,\nu} 
 \gamma^\mu \gamma^\nu \otimes (k_\mu \circ k_\nu)\,,
 \\
  \mathfrak{b}(k,x)&=\sum_\mu  \gahmu \gamu \otimes (x_\mu \circ k_\mu )+ \gamu \gahmu \otimes( k_\mu \circ x_\mu )  \,,\\
  \mathfrak{c}(k,x)&=\sum_{\mu\neq \nu}    \gahmu \ganu \otimes (x_\mu \circ k_\nu ) + \gamu \gahnu \otimes( k_\mu \circ x_\nu ) \,, \\
  \mathfrak{d}(k,x)&=\sum_\mu \gahmu\gahmu\otimes (x_\mu \circ x_\mu)\,, \\
  \mathfrak{e}(k,x)&=\sum_{\mu\neq \nu}  \gahmu\gahnu\otimes (x_\mu \circ x_\nu)\,.\end{align}
\end{subequations}
For the first term one obtains
\begin{salign} \mathfrak{a}(k,x)&= \sum_{\mu,\nu}
  \gamma^\mu \gamma^\nu \otimes k_\mu \circ k_\nu\\
  &=\sum_{\mu,\nu}
  \gamma^\mu \gamma^\nu \otimes \frac12 \Big( k_\mu \circ k_\nu +  k_\nu \circ k_\mu + [k_\mu, k_\nu] _{\circ} \Big)\\
  &=\sum_{\mu,\nu}\gamma^\mu \gamma^\nu \otimes \frac12   ( k_\mu \circ k_\nu)  + \Big(\eta^{\mu\nu}1_V - \frac12\gamma^\nu\gamma^\mu  \Big) \otimes \Big( k_\nu \circ k_\mu + [k_\mu, k_\nu]\Big) \\
  &=\sum_{\mu,\nu} 1_V \otimes \eta^{\mu\nu }k_\mu \circ k_\nu -
  \frac{1}{2}\gamma^\nu\gamma^\mu \otimes [k_\mu, k_\nu] _{\circ}\,\,.
 \end{salign} 
 To get the first two terms in the RHS of \eeqref{Lichnerowicz} one
 renames indices in the last term. The third summand is precisely
 $\mathfrak{d}$ after applying Lemma \ref{thm:lemma3gammas} with
 $\mu=\nu$.  The fourth term is $\mathfrak{e}$, also by Lemma
 \ref{thm:lemma3gammas}.  The sixth and last term in
 \eeqref{Lichnerowicz} come from $\mathfrak{b}$; if one uses
 $\{\gahmu,\gamu\}=0 $ and \eeqref{gammas13}$|_{\mu=\nu}$, after using
 introducing the chirality element:
 \begin{salign}
   \mathfrak{b}(k,x)&=\sum_\mu \gamu \gahmu \otimes( k_\mu \circ x_\mu -x_\mu \circ k_\mu )\\
   &= \sum_\mu (-1)^\mu \gamma \otimes [k_\mu,x_\mu] \qquad (\text{via
     Eq. \ref{gammadef}})\,.
 \end{salign} 
 We now see that the only Gothic letter left unmatched,
 $\mathfrak{c}$, is precisely the fifth term. Indeed, due to Lemma
 \ref{thm:gamma13},
 \begin{align*}
  \mathfrak{c}(k,x)
  &=  \sum_{\mu\neq \nu}    \ganu \gahmu \otimes (x_\mu \circ k_\nu ) + \gamu \gahnu \otimes( k_\mu \circ x_\nu ) \quad\text{(by Eq. \ref{gammas13c})} \\ 
  &=  \sum_{\mu\neq \nu}    \gamu \gahnu \otimes (x_\nu \circ k_\mu ) + \gamu \gahnu \otimes( k_\mu \circ x_\nu ) \quad \text{(index renaming)}
   \\ 
  &= \sum_{\mu\neq \nu}    \gamu \gahnu \otimes  \{ x_\nu,k_\mu\}
  \\ 
  &= \sum_{\mu\neq \nu}  (-1)^{\mu} e_\mu
  \sum_{\alpha < \sigma } (\delta_{\mu\nu\alpha\sigma}  \mtr{sgn}(\nu - \mu) )\gamma^\alpha \gamma^\sigma \otimes  \{ x_\nu,k_\mu\}  \quad \text{(by Lemma \ref{thm:gamma13})}
   \\ 
  &=\frac12  \sum_{\mu,\nu,\alpha,\sigma}  \big[ (-1)^{\mu} e_\mu\mtr{sgn}(\nu - \mu) \mtr{sgn}(\alpha - \sigma)\delta_{\mu\nu\alpha\sigma}   
  \big]\gamma^\alpha \gamma^\sigma \otimes  \{ x_\nu,k_\mu\} 
 \end{align*}
where in the last step we exploited the skew-symmetry 
of the gammas with different indices to annul the 
restriction $\alpha < \sigma$ on the sum 
by introducing $\mtr{sgn}(\sigma-\alpha)$. 
The term in square brackets is $s_{\mu\nu\alpha\sigma}$.
\end{proof}

\begin{table} 
 \begin{tabular}{lcl} \topline   \headcol
   & & \hspace{.4cm}\textsc{Riemannian} \\  
   \headcol \textsc{Concept} &  \textsc{Smooth Geometry}&  \textsc{Fuzzy Geometry}  \\[1ex]  \midline
   Base of $\mathfrak X(U)$ & $  E_\mu= e^i _\mu (x) \partial_i$ &  \hspace{.48cm}$l_\mu= [L _\mu, \balita ]$  \\
   Spin connection & $ \sum_{i,k}\frac{1}{4} \omega^{\phantom{j}}_{[\alpha| ik } e^i_{|\sigma} e^k_{\nu]} $ &    $
                                                                                                              h_{\alpha \sigma \nu}=  \{ H_{\alpha \sigma\nu}, \balita \}$ \\  \bottomlinec
 \end{tabular}
 \vspace{.2cm}
\caption{Analogies between smooth spin geometry
and Riemannian fuzzy geometries\label{tab:analogies}. Local 
expressions in a chart $U$ of $M$ are given. Here, $\mathfrak X(U)$ are
the vector fields on $U$, whose non-coordinate base is $\{E_\mu\}$.}
\end{table}
Notice that the analogy in Table
\ref{tab:analogies} goes further, since in the case of a smooth
manifold spin manifold $(M,g)$, the fields
$\partial_0,\ldots,\partial_3$, or equivalently $E_0,\ldots, E_3$,
(locally) span the space of vector fields $\mathfrak X(M)$ on $M$,
that is, derivations in $C^\infty (M)$.  {The analogue of
  $\partial_j$ is here (after the base change to $E_\mu$) the
  derivation in $\mtr{Der}(\MN)$ that corresponds to
  $l_\mu = \adj_{L_\mu} = [L_\mu,\balita]$.}

\section{Gauge matrix spectral triples} \label{sec:GeneralFinACGeom}

We restrict the discussion to even KO-dimensions ($\epsilon'=1$) and
define the main spectral triples for the rest of the article.  Their
terminology is inspired by the results.  The reader might want to see
Table \ref{tab:notation}, which will be hopefully helpful to grasp the
organization of the objects introduced this section.  But first, we
recall that the spectral triple product $G_1\times G_2$ of two real,
even spectral triples $G_i=(\A_i,\H_i,D_i, J_i,\gamma_i)$
is
\[(\A_1\otimes \A_2 ,\H_1\otimes \H_2, D_1\otimes 1_{\H_2} + \gamma_1
  \otimes D_2, J_1\otimes J_2, \gamma_1\otimes\gamma_2)\,.\]
\begin{definition}
  We define a \textit{gauge matrix spectral triples}  as the spectral triple product
  $G\fuz\times F$ of a fuzzy geometry $G\fuz$ with a finite geometry
  $F=(\A_F,\H_F,D_F, J_F,\gamma_F)$, $\dim \A_F < \infty$.  If $F$ is
  a finite geometry with $\A_F=\Mn$ and $\H_F=\Mn$ with $2\leq n$, we
  say that $G\fuz\times F$ is a \textit{Yang-Mills--Higgs matrix spectral triple}.  If moreover $D_F=0$ above holds,
  then $G\fuz\times F$ is called \textit{Yang-Mills matrix spectral triple}.
\end{definition}

We should denote these geometries by $G\hp{N}\fuz\times F\hp{n}$, 
but for sake of a compact notation, we leave those integers implicit
and write $G\fuz\times F$. 

\subsection{Yang-Mills theory from gauge matrix spectral triples}

In order to derive the SU$(n)$-Yang-Mills theory on a fuzzy base we
choose the following inner space algebra: $\A_F= \Mn$. This algebra
acts on the Hilbert space $\H_F=\Mn$ by multiplication. The Connes'
1-forms $\Omega^1_{D}(\A)$ for $\A=\MN\otimes \Mn$ are then elements
of the form
\begin{align}
 \label{aboveform}
\omega= \sum \ac [D, \cc ] \, \with \, \ac= \sum  W \otimes  a ,\,\,\,  \cc=\sum   T \otimes c \in \MN \otimes \Mn\,,
\end{align}
where the sums are finite. The latter algebra is the fuzzy analogue 
of the algebra $C^\infty(M, \A_F)=C^\infty(M)\otimes \A_F$ of an
($\infty$-dimensional, smooth) almost-commutative geometry.

In order to compute the fluctuated Dirac operator, we start in this
section with the fluctuations along the fuzzy geometry (labeled with
$\fay$) and leave those along the $F$ direction for the Section
\ref{sec:YMH}. Thus, turning off the `finite part' $D_F=0$, one
obtains
\begin{align}\label{DiracFluctuationsGeneral}
\Dg:=D_{\omega\fuz}= D\fuz\otimes 1_F + \omega\fuz + J \omega\fuz J\inv 
\end{align}
for $\omega\fuz$ of the form \eqref{aboveform}, with respect to the `purely fuzzy' Dirac operator
\begin{subequations}\label{DiracWithoutFluct}
\begin{align}
D\fuz\otimes 1_F&= \sum_{\mu } \gamma^\mu \otimes \{\mathsf{K}_\mu, \balita \}_{e_\mu}
+\gahmu \otimes \{\mathsf{X}_\mu, \balita \}_{e_{\hat \mu}}\,,\\
\mathsf{K}_\mu&= K_\mu\otimes 1_F\quad \mtr{and}\quad \mathsf{X}_\mu:= X_\mu\otimes 1_F\,.
\end{align}%
\end{subequations}%

\begin{thm}\label{thm:flucutatedD_AC}
  On the Yang-Mills matrix spectral triple over a
  four-dimensional fuzzy geometry of type $(p,q)$, i.e. of signature
  $\eta=\diag(+_p,-_q) $, the fluctuated Dirac operator
  $D=D\fuz\otimes 1_F$ reads
\begin{align}\label{fullyfluctuatedD_AC} 
\Dg:=D_{\omega\fuz} & =\sum_\mu \gamma^\mu \otimes   \{ \mathsf{K}_\mu + \mathsf{A}_\mu, \balita  \}_{e_\mu} + \gamma^{\hat\mu} \otimes \{ \mathsf{X}_\mu + \mathsf{S}_\mu,\balita\}_{e_{\hat\mu }}\,,\end{align}
in terms of matrices $\mathsf{A}_\mu, \mathsf{S}_\mu \in \Omega^1 _D(\A\fuz \otimes \A_F)$ satisfying 
\begin{align}
 ( \mathsf{A}_\mu)^*&= \vphantom{\int}e_\mu \mathsf{A}_\mu,\qquad \text{and}\qquad ( \mathsf{S}_\mu)^*= (-1)^{q + 1 }e_\mu \mathsf{S}_\mu \,.
\end{align} 
Here, the curly brackets are a generalized commutator $\{A,B\}_{\pm} = AB  \pm BA$ depending on $e_\mu,e_{\hat\mu}\in \{+1,-1\}$.
\end{thm} 

\begin{proof} 
  The theorem follows by combination of Lemma \ref{thm:lemma1gamma}
  with Lemma \ref{thm:lemma3gammas}, both proven below.
\end{proof}

\begin{lemma}[Fluctuations with respect to the $K_\mu$-matrices]\label{thm:lemma1gamma} 
  With the same notation of Theorem \ref{thm:flucutatedD_AC} and
  setting $X_\mu= K_{\hat \mu}=0$---cf. \eeqref{Dirac4d_explicit} and
  \eeqref{minusculas}---the innerly fluctuated Dirac operator $\Dg$ is
  given by
\begin{align}\label{flucutated1gamma} 
 \Dg|_{X=0} & =\sum_\mu \gamma^\mu \otimes \big \{ \mathsf{K}_\mu + \mathsf{A}_\mu, \balita  \big \}_{e_\mu} \, \where e_\mu( \mathsf{A}_\mu)^*= \mathsf{A}_\mu \in \Omega^1 _D(\A)\,. 
\end{align}
\end{lemma}
\begin{proof*}{Proof of Lemma \ref{thm:lemma1gamma}}
We set $X=0$ globally in this proof.
Pick a homogeneous vector in the
full Hilbert space $\Psi=v\otimes Y\otimes \psi \in \H= V \otimes \MN\otimes \H_F$. 
For $\ac=1_V\otimes  W \otimes  a$ and $\ac'=1_V\otimes  T \otimes  c$ parametrized by
$T,W\in \MN$ and $a,c\in \A_F$, the action of $ \omega$ on $\Psi$ yields
\allowdisplaybreaks[3]
\begin{salign}
\omega\fuz (\Psi)& = \ac [D\fuz\otimes 1_F, \ac' ] (\Psi) 
\\[3pt] & = (1_V\otimes  W \otimes  a) \big [
\sum_\mu
\gamma^\mu 
\otimes \{K_\mu  , \balita \}_{ e_\mu} \otimes 1, 1_V \otimes T \otimes c \big ]  (\Psi)\\[-3pt]
&=\sum_\mu\gamma^\mu v \otimes W\Big( 
\{K_\mu  , \balita \}_{ e_\mu} T - T \{K_\mu  , \balita \}_{ e_\mu} 
\Big) (Y \otimes a c \psi ) \\
&=\sum_\mu\gamma^\mu v \otimes W\Big( 
\{K_\mu  , T Y \}_{ e_\mu} - T \{K_\mu  , Y \}_{ e_\mu} 
\Big)  \otimes  a c  \psi  \\
&=\sum_\mu\gamma^\mu v \otimes W\Big( 
 K_\mu   T Y +  e_\mu  T Y K_\mu  - T (K_\mu Y +{ e_\mu} Y K_\mu )
\Big)  \otimes  a c  \psi  \\
&=\sum_\mu\gamma^\mu v \otimes W\big( 
 [ K_\mu ,  T ]\big) Y \otimes  a c  \psi  \\ & = \sum_\mu\big( \gamma^\mu \otimes W [K_\mu, T] \otimes a c \big) \Psi 
\end{salign}
so $\omega\fuz = \sum_\mu \gamma^\mu \otimes A_\mu \otimes b $, relabeling $b=a c \in \A_F$ and $A_\mu: = W [K_\mu, T]$.  
Notice that since 
\begin{align}
\label{restricting_b}
 (\gamma^\mu \otimes A_\mu \otimes b)^* &= e_\mu
\gamma^\mu \otimes   A_\mu ^* \otimes b^* 
\end{align}
the self-adjointness condition $\omega\fuz^*=
\omega\fuz$ is achieved if and only if $(A_\mu\otimes b)^*=e_\mu (A_\mu\otimes b) $ 
for each $\mu$. The second part of the inner fluctuations is, 
for each 
\[\Psi=v\otimes Y\otimes \psi \in V \otimes \MN\otimes \Mn\,, \]
the next expression:
\begin{salign}
(J \omega\fuz J\inv ) (\Psi)& = \vphantom{\sum_\mu}\big( J \ac [D\fuz\otimes 1_F, \ac' ] J\inv \big)  (\Psi) \\ & = \sum_\mu (C\otimes *_N \otimes *_n )
(\gamma^\mu C\inv v \otimes A_\mu Y^* \otimes b \psi ^*)
\\ & =
\sum_\mu  
(\underbrace{C \gamma^\mu C\inv}_{\gamma^\mu} v \otimes (A_\mu Y^*)^* \otimes  (b \psi ^*)^* \qquad \mbox{(cf. Eq. \ref{C})} 
\\ &= 
\sum_\mu  
 \gamma^\mu  v \otimes YA_\mu^* \otimes   \psi b^*
\\ &= 
\sum_\mu 
( \gamma^\mu \otimes 1_{\MN}\otimes 1_n) \Psi
(1_V\otimes A_\mu \otimes b)^*
\\ &= 
\sum_\mu  
(e_\mu \gamma^\mu \otimes 1_{\MN}\otimes 1_n) \Psi
(1_V\otimes A_\mu \otimes b)\,,
\end{salign}
where the last step is a consequence of Eq. \eqref{restricting_b}. Thus
$ J \omega\fuz J\inv  = \sum_\mu  e_\mu \gamma^\mu \otimes (\balita ) (A_\mu \otimes b)$ 
where the bullet stands for the argument in $\MN\otimes \Mn\subset \H$ 
to be multiplied by the right. Hence%
\begin{subequations}%
\label{flucutationsYM}%
\begin{align}%
\omega\fuz +  J \omega\fuz J\inv  &= \sum_\mu  \gamma^\mu \otimes \Big (  
A_\mu \otimes b + e_\mu  (\balita)
(A_\mu \otimes b) \Big)\,, \with \qquad\\  e_\mu (A_\mu\otimes b)^* 
&=  A_\mu \otimes b \in \Omega^1 _D(\A) \mbox{ for each $\mu$.} 
\label{flucutationsYMb}
 \end{align}%
\end{subequations}%
As a result, the fully-fluctuated operator acting on 
$\Psi=v  \otimes Y\otimes \psi \in \H$ is
\begin{salign}
D_{\omega\fuz} \Psi & = \overbrace{\sum_\mu \gamma^\mu v \otimes \{K_\mu, Y\}_{ e_\mu} \otimes \psi}^{(D\fuz \otimes 1_F) \Psi}   + \sum_\mu \gamma^\mu v \otimes \big( A_\mu Y \otimes  b \psi +  e_\mu Y A_\mu \otimes \psi b\big) \,.
\end{salign}
or defining $\mathsf{K}_\mu := K_\mu \otimes 1_n $ and $\mathsf{A}_\mu:= A_\mu \otimes b \in \MN \otimes \Mn$, one has
\begin{align}  
D_{\omega\fuz} & =\sum_\mu \gamma^\mu \otimes \big \{ \mathsf{K}_\mu + \mathsf{A}_\mu, \balita  \big \}_{e_\mu}\,, \qquad ( \mathsf{A}_\mu)^*= e_\mu \mathsf{A}_\mu \in \Omega^1 _D(\A)\,. 
\qedhere
\end{align}
\end{proof*}
The triviality of the part of the Dirac operator 
along the finite geometry $F$ implies that 
\[ \Omega^1_{D}(\A)= \Omega^1 _{D\fuz \otimes 1_F}\big( \MNn \big)=
  \Omega^1 _{D\fuz}(\MN) \otimes \Mn\,, \] where $\MNn$ abbreviates
$\MN \otimes \Mn$ (in sub-indices, later further shortened as
$M_{N\otimes n}^{\C}$ too), and the significance of each factor can be
obtained by comparison with the smooth case. There, the inner
fluctuations of a Dirac operator on an almost-commutative geometry are
given by
\[\sum_\mu\Gamma^\mu \otimes (\mathbb{A}_\mu - J_F \mathbb{A}_\mu J_F ),\,\with   \mathbb{A} _\mu =-
  \ii a \partial_\mu {b} \in C^\infty (M) \otimes \A_F\, .\] Recall
that in the smooth case it is customary to treat only Riemannian
signature together with self-adjointness (which we do not assume) for
each gamma-matrix $\Gamma^i= \mathrm{c}(\dif x^i)$, $\mathrm{c}$ being
Clifford multiplication. For each point $x$ of the base manifold $M$
one has
$\mathbb{A}_i(x)\in \ii\, \mathfrak{su}(n)= \ii\, \mtr{Lie \,SU}(n)$.
Since Eq. \eqref{flucutationsYMb} represents the fuzzy analogue, that
equation can be further reduced to
$b^*=b\in \Mn_{\mtr{s.a.}}=\ii \,\mathfrak{u}(n) $ and
$(A_\mu)^*=e_\mu A_\mu $, that is
\begin{align}
 A_\mu \in \begin{cases}
           \ii \,\mtf{u} (N)  & \text{if }  e_\mu= +1   \text{ iff }  (\gamma^\mu)^*= + \gamma^\mu \,, \\
            \mtf{u} (N)  & \text{if } e_\mu= -1    \text{ iff }  (\gamma^\mu)^*= -\gamma^\mu  \,.
           \end{cases}
\end{align}
We now have to add the fluctuations resulting from the 
triple products of gamma matrices. 

\begin{lemma}[Fluctuations with respect to the $X_\mu$-matrices]
\label{thm:lemma3gammas} 
With the same notation of Theorem \ref{thm:flucutatedD_AC} and
additionally setting $K_\mu=0$, the innerly fluctuated Dirac operator
$\Dg$ is given by
\begin{align}\label{flucutated2gamma} 
 \Dg|_{K=0} & =\sum_\mu \gamma^{\hat \mu} \otimes \big \{ \mathsf{X}_\mu + \mathsf{S}_\mu, \balita  \big \}_{e_{\hat \mu}}\,, \,\,\,  (-1)^{q+1} e_\mu( \mathsf{S}_\mu)^*= \mathsf{S}_\mu \in \Omega^1 _D(\A)\,. 
\end{align}

\end{lemma}

\begin{proof*}{Proof}
See Appendix \ref{app:prooflemma}. 
\end{proof*}
From the last subsections, the rules for $S_\mu$ and $A_\mu$
lead to the manifest self-adjointness of $D_{\omega\fuz}$.

\begin{table}


\includegraphics[width=\textwidth]{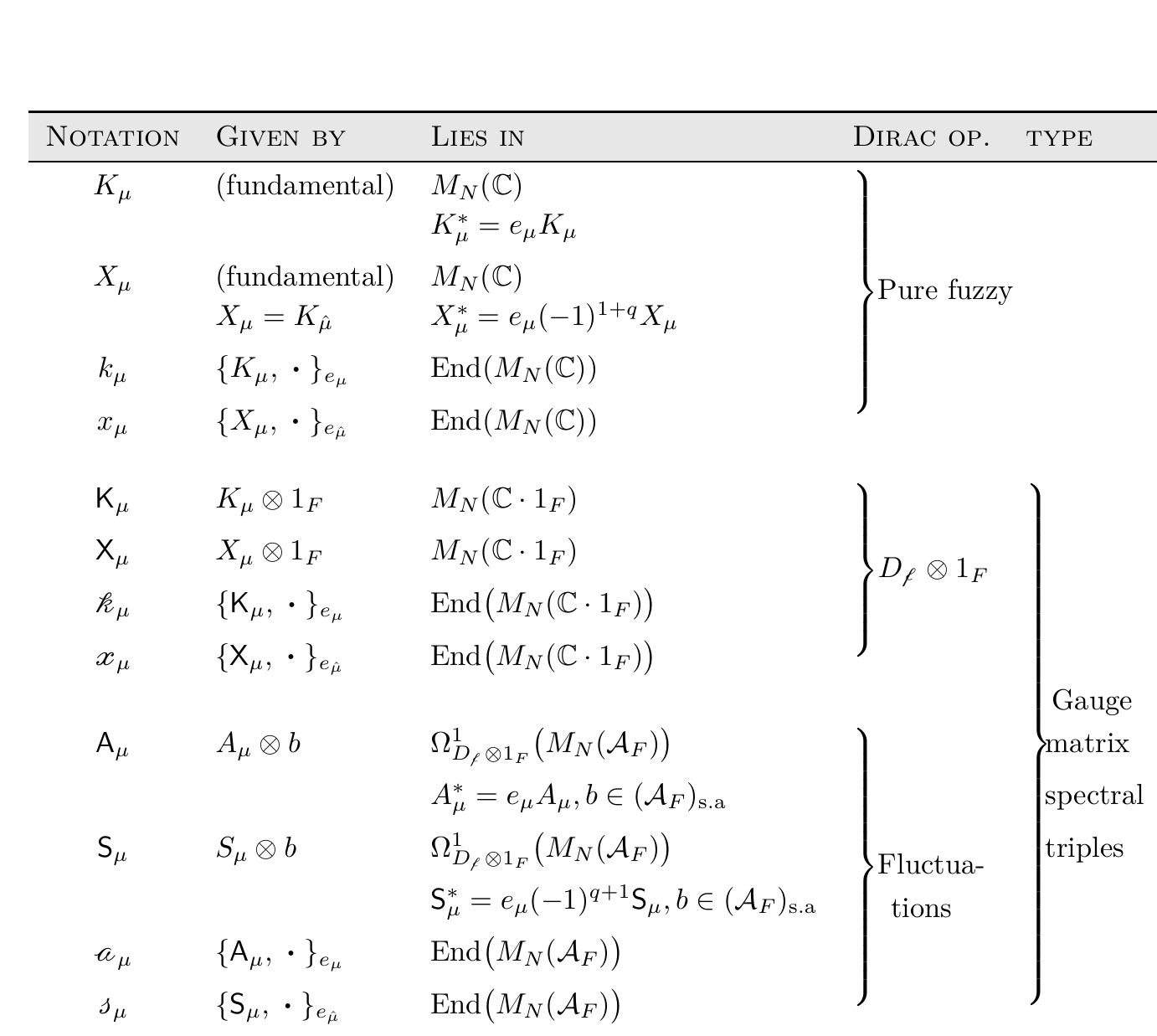}

 \caption{
 Notation for different matrices and operators 
 appearing in the Dirac operator $D=D\fuz\otimes 1_F$ of 
  $G\fuz\times F $ (case with $D_F=0$)
 and its fluctuations. In the table  $M_N(\C\cdot 1_F)=\MN\otimes ( \C \cdot 1_F)$.}
\label{tab:notation}\end{table}

\subsection{Field Strength and the square of the fluctuated Dirac operator}
We introduce now the main object of the gauge theory.  To this end,
let\footnote{Here we emphasize the composition to avoid potential
  confusion arising from the objects inside commutator already being
  (anti-)commutators themselves. Should no confusion arise, so we drop
  the $\circ$.}
$[\mathsf f,\mathsf g]_{\circ}= (\mathsf f\circ \mathsf g)-(\mathsf
g\circ \mathsf f)$ for any endomorphisms $\mathsf f,\mathsf g$ of the
same vector space. Similarly, we define
$\{\mathsf f,\mathsf g\}_\circ= (\mathsf f\circ \mathsf g) + (\mathsf
g\circ \mathsf f)$.
\begin{definition}\label{def:FieldStr} We abbreviate the following (anti)commutators 
\begin{subequations}
\begin{align}
\km_\mu&:= \{\mathsf{K}_\mu,\balita \}_{e_{\mu}}\,, &&&
\xm_\mu&:= \{\mathsf{X}_\mu,\balita \}_{e_{\hat\mu}}=\{\mathsf{K}_{\hat \mu},\balita \}_{e_{\hat\mu}}\,, \\
\am_\mu&:=\{ \mathsf{A}_\mu,\balita  \}_{e_\mu}\,, &&&
\sm_\mu&:=\{ \mathsf{S}_\mu,\balita  \}_{e_\hmu}\,.
\end{align}
\end{subequations}
It follows in particular that 
$k_\mu \otimes 1_F= \km_\mu.$
The \textit{field strength} $\mathscr{F}_{\mu\nu}\in \mtr{End}( \A\fuz\otimes \A_F) $ of a gauge matrix spectral triple $G_\fay \times F$ is defined as
\begin{align}\label{FieldStrength}
 \mathscr{F}_{\mu\nu}:=[\day_\mu,\day_\nu] = [\km_\mu + \am_\mu, \km_\nu+\am_\nu]_{\circ}\,
 \end{align}
 where $\day_\mu:=\km_\mu + \am_\mu$.
\end{definition}

\begin{proposition} \label{thm:flatWeitzenb} The square of the
  fluctuated Dirac operator of a Yang-Mills gauge matrix
  spectral triple that is flat ($X=0$, $\xm=\sm=0$) is given by
\begin{align}\label{flatWeitzenb} 
\Dg^2|_{X=0}=  \frac12  \sum_{\mu,\nu}\gamma^\mu\gamma^\nu \otimes  \mathscr{F}_{\mu\nu}+ 1_V\otimes \vartheta\,,
\end{align}
where
\begin{align}\label{Theta}
 \vartheta := 
  \sum_{\mu,\nu}\eta^{\mu\nu } (\am_\mu +\km_\mu ) \circ (\am _\nu+\km _\nu)\,. \end{align}
\end{proposition}%
\begin{proof}
Squaring $\Dg   = D\fuz \otimes 1_F +\omega\fuz + J\omega\fuz J\inv$ one gets
\begin{align} \label{referenciatrivial}
\Dg^2&= (D\fuz)^2 \otimes 1_F + (D\fuz \otimes 1_F )(\omega + J\omega J\inv)  \\ & + (\omega + J\omega J\inv) (D\fuz  \otimes   1_F)  +  (\omega + J\omega J\inv ) ^2 \,. \nonumber
\end{align}
The first summand is known from Lemma \ref{thm:Lichnerowicz}.  One
obtains the last summand by \eeqref{flucutated1gamma} and using the
Clifford algebra relations just as in the proof of that lemma. The
result reads
\begin{salign}
 (\omega + J\omega J\inv ) ^2 |_{X=0}\label{omega2}\numerada
 &= \sum_{\mu,\nu} \gamma^\mu \gamma^\nu \otimes \frac12   ( \am_\mu \circ \am_\nu)   \\ & + \sum_{\mu,\nu}\Big(\eta^{\mu\nu}1_V - \frac12\gamma^\nu\gamma^\mu  \Big) \otimes \Big( \am_\nu \circ \am_\mu + [\am_\mu, \am_\nu]\Big) \\
 &=\sum_{\mu,\nu} 1_V \otimes \eta^{\mu\nu }\am_\mu \circ \am_\nu  + \frac{1}{2}\gamma^\nu\gamma^\mu \otimes [\am_\mu, \am_\nu] \,.  
 \end{salign}
 being $[f,g]=f\circ g - g\circ f $ a simplified notation for the
 composition-commutator.  We renamed indices and rewrote the last
 summand in \eqref{omega2} as
 $\frac{1}{2}\gamma^\mu\gamma^\nu \otimes [\am_\mu, \am_\nu] $ To make
 the notation lighter, we also mean by $\am_\mu \am_\nu$ the
 composition $\am_\mu \circ \am_\nu$ from now on (also for
 $\kay_\mu$).  Using \eeqref{flucutated1gamma} one can obtain for the
 two summands in the middle of \eeqref{referenciatrivial}; further
 abbreviating $\km_\mu = \{ K_\mu \otimes 1_F , \balita \}_{e_\mu} $
 one obtains
 \allowdisplaybreaks[3]\begin{salign}\vphantom{\frac{1}{2}} &
   \big\{(D\fuz \otimes 1_F )(\omega + J\omega J\inv) + (\omega +
   J\omega J\inv) (D\fuz \otimes 1_F) \big\} |_{X=0}
   \\
   \vphantom{\frac{1}{2}} &=\sum_{\mu,\nu} \big(\gamma^\nu \otimes \{
   \mathsf{K}_\nu , \balita \}_{e_\nu} \big) \big(\gamma^\mu \otimes
   \{ \mathsf{A}_\mu , \balita \}_{e_\mu}\big ) + \big(\gamma^\mu
   \otimes \{ \mathsf{A}_\mu , \balita \}_{e_\mu}\big) \big(\gamma^\nu
   \otimes \{ \mathsf{K}_\nu , \balita \}_{e_\nu}\big)
   \\[-2pt]
   &=\sum_{\mu,\nu} (\gamma^\nu \otimes \km_\nu )(\gamma^\mu\otimes \am_\mu )
   +(\gamma^\mu\otimes \am_\mu )(\gamma^\nu \otimes \km_\nu
   )\vphantom{\frac{1}{2}}
 \\ 
&= \sum_{\mu,\nu} (\gamma^\nu \gamma^\mu \otimes \km_\nu \am_\mu)
+ (\gamma^\mu \gamma^\nu \otimes \am_\mu  \km_\nu) \vphantom{\frac{1}{2}}
\\
&= \sum_{\mu,\nu} \frac{1}{2}\Big( 2\eta^{\mu\nu} 1_V- \gamma^\nu \gamma^\mu \Big)
\otimes (\km_\mu \am_\nu + \am_\mu \km_\nu )
+
 \frac{1}{2} \gamma^\mu \gamma^\nu \otimes (\km_\mu \am_\nu + \am_\mu \km_\nu) 
 \\
 &=\sum_{\mu} 1_V \otimes ( \km_\mu \am^\mu +
 \am^\mu \km_\mu) 
 +\frac12 \sum_{\mu,\nu}\gamma^\mu\gamma^\nu 
 \otimes
 \Big( -\km_\nu \am_\mu - \am_\nu \km_\mu +\km_\mu \am_\nu + \am_\mu \km_\nu \Big)
 \\
 &= \sum_{\mu}1_V\otimes  
 \{\km_\mu ,\am^\mu\}  
  +\frac12\sum_{\mu,\nu} \gamma^\mu\gamma^\nu 
 \otimes \Big( 
 [\km_\mu,\am_\nu ] - [\km_\nu,\am_\mu] 
 \Big)\,. \label{crossedterms} \numerada
\end{salign}
Again, we used the Clifford relations for the gamma matrices and
renamed indices. 
Equations \eqref{omega2} and \eqref{crossedterms}
imply  
\begin{align*}
 (\Dg|_{\xm=\sm=0})^2&= 
D\fuz^2\otimes 1_F 
  +\frac12\sum_{\mu,\nu} \gamma^\mu\gamma^\nu \otimes\Big(
  [\km_\mu,\am_\nu ] - [\km_\nu,\am_\mu] + [\am_\mu, \am_\nu]
  \Big) \\[3pt] & + \sum_{\mu}1_V\otimes  \am^\mu \am_\mu  
 +\{\km^\mu, \am_\mu  \}\,.
\end{align*}
Expanding \eeqref{FieldStrength} 
\begin{align} \label{FieldStrFull}
 \mathscr{F}_{\mu\nu}= [\km_\mu , 
\am_\nu]_{\circ} -  [\km_\nu , \am_\mu]_{\circ}
+ [\am_\mu,  \am_\nu]_{\circ} +[ \km_\mu,\km_\nu]\,,
\end{align}
and using Lemma \ref{thm:Lichnerowicz}$|_{X=0}$ 
together with $\vartheta=   \km^\mu \circ \km_\mu + \{ \km^\mu , \am_\mu  \}_{\circ}  +\am^\mu \circ \am_\mu  $ one gets 
 the result.\end{proof}

\allowdisplaybreaks[2]
\begin{proposition}\label{thm:Weitzenb}
The fluctuated Dirac operator of a finite Yang-Mills geometry 
satisfies 
\begin{align}  \nonumber 
\Dg^2 &=
  \sum_{\mu,\nu}1_V \otimes \eta^{\mu\nu } (\km_\mu+\am_\mu) \circ(\km_\nu+\am_\nu)  + \frac{1}{2}\gamma^\mu\gamma^\nu \otimes [\km_\mu+\am_\mu, \km_\nu+\am_\nu]_\circ
 \\ 
\nonumber
  &
+
\sum_\mu   \det(\eta)  ( - e_\mu) 1_V\otimes (\xm_\mu + \sm_\mu) \circ (\xm_\mu + \sm_\mu )
\\ 
& \label{Weitzenb} 
+ \sum_{\mu < \nu } 
t_{\mu\nu} \gamma^{\mu} \gamma^\nu \otimes [\xm_\mu + \sm_\mu,\xm_\nu + \sm_\nu]_{\circ}
\\  
&  \nonumber 
+\frac12\sum_{\mu,\nu,\sigma,\alpha } s_{\mu\nu\alpha\sigma} \cdot \gamma^\alpha \gamma^\sigma \otimes \{ \xm_\nu + \sm_\nu,\km_\mu +\am_\mu\}_{\circ}
\\ \nonumber
&
+ \sum_{\mu}(-1)^\mu \gamma \otimes [\xm_\mu + \sm_\mu,\km_\mu +\am_\mu]_{\circ}\,.
\end{align}

\end{proposition}
\begin{proof}
According to \eeqref{fullyfluctuatedD_AC} 
\[D_\omega  =\sum_\mu \gamma^\mu \otimes (\km_\mu +\am_\mu )+ \gamma^{\hat\mu} \otimes  (\xm_\mu+\sm_\mu)
\]   
so $D_\omega^2$ has the same structure already observed in the `fuzzy
Lichnerowicz formula' above (Lem. \ref{thm:Lichnerowicz}). To be
precise, notice that one can compute the square of the present Dirac
operator by replacing the in $D\fuz$ the following operators:
$k\to \km + \am $ and $x\to \xm + \sm$.
\end{proof}

\subsection{Gauge group and gauge transformations}
For any even spectral triple, the Hilbert space $\H$ is an
$\A$-bimodule.  The right action of $\A$ on the Hilbert space
$\mathcal H \ni \Psi$ is implemented by the real structure $J$,
$ \Psi a:= a^{\mtr{o}} \Psi := J a^* J\inv \Psi $.  Both actions
define the adjoint action $\Adj(u) \Psi := u \Psi u^*$ of the
unitarities $u\in \mathcal U(\A)$ on $\H$.  We want to determine the
action of the unitarities $\mtc U (\A)=\{u \in \A \mid u^*u=1=uu^* \}$
of the algebra $\A$ on the Dirac operator,
\begin{subequations}
 \label{GaugeD}
 \begin{align}
 U \big(D + \omega + \epsilon' J \omega J\inv \big) U^*
&= D + \omega_u + \epsilon' J
\omega_u J\inv\,, \\ \quad U:\!&=\Adj_u,\, u\in \mtc U(\A)\,,
\end{align}
which namely leads to the transformation rule  
\begin{align}
\omega \mapsto \omega_u = u \omega u^* +u[D,u^*]
\end{align}
\end{subequations}%
for the inner fluctuations. It is instructive to present a variation
of the original proof given in \cite[Prop. 1.141]{ConnesMarcolli} for
the analogous property of general spectral triples. Verifying this
again is important, since the axiom $[a,b^{\mtr o}]=0$ that appears in
\textit{op. cit.}, does not appear in the present axioms. However,
according to Remark \ref{rmk:commrel} above, it is a consequence of
these in the fuzzy setting. So one can see that not only there, but
also for gauge matrix spectral triples, the commutant property
$[\A,\A^{\mtr o}]=0 $ (elsewhere an axiom) holds.  Indeed, since
$J=C\otimes *_N \otimes *_n$, for $a,b\in\A$,
\begin{align}\label{commutantprop}
 a b^{\mtr o} (v\otimes m )  = 
 a J  ( v \otimes b^* m^* )&= 
  v \otimes (a  m b  ) 
 \\ 
 &= b^{\mtr o}  a (v\otimes m)\,,  \quad v \in V, m\in \MNn \,. \nonumber
\end{align}
The commutant property is essential for the
subalgebra \begin{align}\label{subalgebraJ} \A_{J}:=\{ a \in \A \mid a
  J = J a^*\} \subset \A
               \end{align}
               to be also a subalgebra of the center $Z(\A)$, as we
               will see later.

\begin{proof*}{Proof of \eeqref{GaugeD}; adapted from
                   \cite[\S 10]{ConnesMarcolli} to fuzzy geometries}
We split the adjoint action into the right action by $u^*$, $z:=(u^*)^{\mtr o}=J u J\inv $, and the left action by $u$, $U=u z$. 
\begin{itemize}\setlength\itemsep{.4em}
 \itemb Transformation of $D$: Applying 
$wD w^* = D + w [D,w^*]$ 
consecutively for 
$w=z,u$, one gets 
\begin{align}\label{above}
  U D U^* = u ( D + z[D,z]) u^*  = D + u[D,u^*] + z[D,z^*]\,.
\end{align}
\itemb Transformation of $\omega$: since $\omega \in \Omega^1_D(\A)$,
$\omega= a[D,b]$ (or sum of this 1-forms), one also has 
\begin{align*} 
\omega z^* = \omega u^{\mtr{o}} \stackrel{\eqref{orderone}}{=} a[D,b] u^{\mtr{o}}
=a  u^{\mtr{o}}[D,b]  \stackrel{\eqref{RLcommute}}=  u^{\mtr{o}} a [D,b] =   z^* \omega\,,
\end{align*}
so  $ U \omega U^* = u (z  \omega  z ^*)u  = u \omega u^* $, since $z z ^*=1$. 
Also the term  $u[D,u^*] $ is absorbed from the pure Dirac operator, then \[
\omega \mapsto  u \omega u^*  + u[D,u^*]\,.
\]

\itemb Transformation of $ J\omega J\inv $: Similarly one obtains
$U J \omega J\inv U^* = J (u \omega u^* ) J$. But actually
$ z[D,z^*] $ from Eq. \eqref{above} can be taken from the
transformation of the pure Dirac operator and passed to that of
$J \omega J\inv$, contributing, by the axioms \eqref{signos} of the
fuzzy geometry, since one can rearrange it as
\[
z[D,z^*] = J u J\inv \big( D J u^* J\inv - J u^* J\inv D \big)= \epsilon' J u [D, u^*] J\inv \,. \qedhere
\]
\end{itemize}
\end{proof*}

The gauge group $ \gauge$ of a real spectral triple is defined via the
adjoint action $\Adj_u(a) = u a u^*$ of the unitary group $\U(\A)$ on
$\H$ as follows:
\begin{equation} \label{gauge}
\gauge(\A,J) = \{ \Adj_u \,|\, u\in \U(\A) \} = \{  uJ uJ\inv \,| \,u \in \U(\A) \}\,.
\end{equation}
Before proceeding to compute it for a case concerning our study,
we do the notation more symmetric, setting $n_1=N$ and $n_2=n$
for the rest of this section. We assume $n_1 > n_2 \geq 2$.
The next statement is not surprising, but due to the presence of the tensor product, some care is needed. 
\begin{proposition} \label{thm:gaugegroup} Let
  $G_1\times G_2=G\fuz\times F$ be a gauge matrix
  geometry, with algebra $\A=A_1\otimes A_2$, $A_1=M_{n_1}(\C) $ and
  $ A_2 = M_{n_2}(\C)$, and reality $J=J_1\otimes J_2$.  The gauge
  group is given by the product of unitary projective groups
  $\gauge(\A,J) = \mathrm {PU}(n_1) \times \mathrm {PU}(n_2)$.
\end{proposition}
   
Before proving this proposition, broken down in some lemmata below, we
recall the characterization of the gauge group that will be
used. Namely, the next short sequence is exact, according to
\cite[Prop. 6.5]{WvSbook}:
\begin{align}
\label{SESgauge} 
1\to \mtc U(\A_J) \to \mtc U (\A) \to \gauge (\A,J)\to 1\,.
\end{align}
Thus, if the groups $\mtc U(Z(\A))$ and $\mtc U(\A_J)$ coincide,
then \begin{align}\label{ifalgebrascoincide} \gauge (\A,J) \cong \mtc
  U(\A) / \mtc U(Z(\A)) \,.
     \end{align} We now verify that they do,
     so that after computing 
     $\mtc U(\A) $ and $\mtc U(Z(\A))$, we can finally obtain the gauge 
     group by this isomorphism \eqref{ifalgebrascoincide}.

\begin{lemma} For $\A$ and $J$ as in Proposition \ref{thm:gaugegroup},
$\mtc U(Z(\A)) =\mtc U(\A_J)$.
\end{lemma}
\begin{proof}
First, observe that if $a\in \A_J$ and $b \in \A$, then
\[a b= Ja^* J\inv b= a^{\mtr 0 }b = b a^{\mtr 0 } = b Ja^* J= b
  a\,, \] where one gets the equalities at the very left or very right
by the defining property \eqref{subalgebraJ} of $\A_J$, and the third
equality by the commutant property \eqref{commutantprop}. Hence
$\A_J \subset Z(\A)$, and thus $\mtc U(\A_J) \subset \mtc
U(Z(\A))$. \par
We only have to prove the containment
$\mtc U(Z(\A)) \subset \mtc U(\A_J)$.  According to Lemma
\ref{thm:CenterOfTensorProd} (proven in Appendix
\ref{app:prooflemma}),
$Z(\A)=Z(A_1\otimes A_2)= Z(A_1)\otimes Z (A_2)$.  Since the
representation $\varrho$ of $ A_1 \otimes A_2 $ on
$ \H_1\otimes \H_2=V\otimes A_1 \otimes A_2$ is the fundamental on
each factor (except the trivial action on spinor space factor $V$) by
Schur's Lemma, each $ Z(A_i)$ consists of multiples of the identity.
Then, for any $z_1\otimes z_2 \in Z(A_1)\otimes Z(A_2)$ one
has \begin{align*} \varrho[( z_1\otimes z_2)^* ]J \Psi& = (1_V\otimes
  \bar z_1\otimes \bar z_2) (C\otimes *_1\otimes *_2) \Psi \\& =
  (C\otimes *_1\otimes *_2) (1 _V\otimes z_1\otimes z_2 )\Psi =J
  \varrho(z_1\otimes z_2 )\Psi
\end{align*}
where $*_i$ is the involution of $A_i$ and $\Psi$ an arbitrary vector
in the Hilbert space described above.  Therefore
$z_1\otimes z_2 \in \A_J$. One verifies that this proof leads equally
to $Z(\A) \subset \A_J$ by taking other representing element
$z_1\lambda \otimes z_2 \lambda\inv$ $(\lambda\in\C^\times)$ the same
conclusion $Z(\A) \subset \A_J$ is reached, which restricted to the
unitarities gives $\mtc U(Z(\A)) \subset \mtc U(\A_J)$.
\end{proof}

\begin{lemma} \label{thm:SESUAA}
The following is a short exact sequence of groups:
\[
1 \to  \C^\times \to \{\re^+ \times \mtr U (n_1)\} \times_{|\det|} \{\re^+ \times \mtr{U}(n_2)\} \stackrel{\alpha}{\to} \mtc U (A_1\otimes A_2) \to 1\,,
\]
where\footnote{That group is isomorphic to $\re^+ \times \mtr U (n_1) \times \mtr U (n_2)$, but we will keep the full notation and the embedding for later convenience.}
\begin{align}  
\{\re^+ \times \mtr U (n_1)\} &\times_{|\det|} \{\re^+ \times \mtr{U}(n_2)\}
\\
&:= \big\{ (\rho_1,u_1,\rho_2,u_2) \in 
\re^+ \times \mtr{U}(n_1) \times 
\re^+ \times \mtr{U}(n_2) \mid \rho_1\rho_2 =1 \big\}\,. \nonumber\end{align}
\end{lemma}
\begin{proof}
  Let us abbreviate the group in the middle as follows
  $G= \{\re^+ \times \mtr U (n_1)\} \times_{|\det|} \{\re^+ \times
  \mtr{U}(n_2)\}$ and define $\alpha : G\to \mtc U (A_1\otimes A_2)$
  by $(\rho_1,u_1,\rho_2,u_2) \mapsto \rho_1u_1\otimes \rho_2u_2$.
  Suppose $(\rho_1,u_1,\rho_2,u_2)\in \ker \alpha$, so that
  $\alpha(\rho_1,u_1,\rho_2,u_2) =\rho_1 u_1\otimes \rho_2 u_2
  =1_{n_1} \otimes 1_{n_2}$.  Since $ \ker \alpha \subset G$ one has
  $\rho_1\rho_2=1$. Thus, the previous equation yields
  $u_1\otimes u_2=1_{n_1} \otimes 1_{n_2}$, which says that in the lhs
  $u_1$ and $u_2$ are a scalar multiples and mutual inverses. Then
  $\ker \alpha \cong \{\rho,\lambda,\rho\inv,\lambda\inv\} $, and if
  one embeds $\C^\times \hookrightarrow G $ as follows (which will be
  the definition of the leftmost map)
  $z=|z|\cdot \ee^{\ii \theta}= r \cdot \ee^{\ii \theta} \mapsto (r,
  \ee^{\ii \theta} ,r\inv, \ee^{-\ii \theta} )$, one gets exactness at
  $G$. \par
  The rightmost map $\mtc U (A_1\otimes A_2) \to 1$ is the determinant
  in absolute value. Its kernel has elements
  $g_1\otimes g_2 \in \mtc U (A_1\otimes A_2)$ satisfying
  $|\det( g_1\otimes g_2)|=1$.  But this condition is satisfied by all
  elements $g_i= \rho_i u_i $, as far as
  $(\rho_1,u_1,\rho_2,u_2 )\in G$. Conversely, if
  $ g_1\otimes g_2 \in \mtc U (A_1\otimes A_2)$ satisfies
  $|\det (g_1\otimes g_2)|=1$, then there exists a
  $\lambda = |\lambda| \cdot \ee^{\ii \psi} \in \C^\times$ with
  $g_1= \lambda \cdot u_1$ and $g_2=\lambda\inv \cdot u_2$. Then
  $ \alpha (|\lambda|, \ee^{\ii \psi} \cdot u_1 , |\lambda|\inv,
  \ee^{\ii \psi} \cdot u_2) =g_1\otimes g_2$. Hence
  $\ker ( |\det (\balita ) |) \subset \im \alpha$ too, and the
  sequence is exact also at $\mtc U (A_1\otimes A_2)$.
 \end{proof}

\begin{figure}
 \includegraphics[width=.45\textwidth]{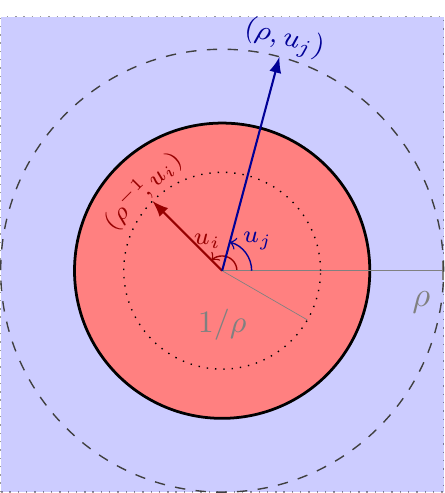}
 \caption{\label{fig:gaugegroup} Illustration of the group $G$. Such group appears in 
 the description of  $\mathcal{U}(A_1\otimes A_2)\cong G / \C^\times$.
There the indices refer to each $ \mtr{U}(n_i)$-factor, $i\neq j$,
 and the $\rho$ and $\rho\inv$ might lie outside the unit circle (thick line).}
\end{figure}

\begin{lemma} \label{thm:UofCenterOfTensorProd}
The following group sequence is exact:
\[
1\to \C^\times \to\C^\times\times_{|\det|}
\C^\times \to \mtc U \{Z(A_1\otimes  A_2)\} \to 1\,.
\]
The group $\C^\times\times_{|\det|}
\C^\times$ in the middle is 
the subgroup of $(\C^\times)^2$
whose entries $(z_1,z_2)$ satisfy
$|z_1|=|z_2|\inv$.
\end{lemma}
\begin{proof}
  The embedding
  $\C^\times \hookrightarrow \C^\times\times_{|\det|} \C^\times $ is
  given by $\lambda\mapsto (\lambda,\lambda\inv)$ and the next map
  $\C^\times\times_{|\det|} \C^\times \to \mtc U \{Z(A_1\otimes
  A_2)\}$ by $(z_1,z_2) \mapsto z_1\otimes z_2$.  Being the rest an
  easier case than that the proof of Lemma \ref{thm:SESUAA}, the
  details on exactness can be deduced from there.
\end{proof}

We are now in position to give the missing proof.

\begin{proof*}{Proof of Proposition \ref{thm:gaugegroup}.}
According to \eeqref{ifalgebrascoincide}, 
\begin{align*}
  \gauge(\A,J) & \cong {\mtc U (A_1\otimes A_2)}\big /{ \mtc U \{Z(A_1\otimes  A_2)\}} \\ &\cong {\mtc U (A_1\otimes A_2)}\big /{ \mtc U \{Z(A_1)\otimes  Z(A_2)\}}\,.\end{align*}
where one passes to the second line by Lemma \ref{thm:CenterOfTensorProd}.
By Lemmas \ref{thm:UofCenterOfTensorProd} for the group in the `numerator' and Lemma 
\ref{thm:SESUAA} for the one in the `denominator',
\[\gauge(\A,J)\cong  
 \frac{
\big[\{\re^+ \times \mtr U (n_1)\} \times_{|\det|} \{\re^+ \times \mtr{U}(n_2)\}\big]\big/\C^\times 
}{ 
 \big(\C^\times\times_{|\det|}
\C^\times\big) /\C^\times
 }
\,.
\]
By the third group isomorphism theorem, one can `cancel out' the
$\C^\times$, and get
\[\gauge(\A,J)\cong  
 \frac{
\big[\{\re^+ \times \mtr U (n_1)\} \times_{|\det|} \{\re^+ \times \mtr{U}(n_2)\}\big]
}{ 
 \big(\C^\times\times_{|\det|}
\C^\times\big)
}
\,.
\]
Notice that in each group $|\det|$ only constrains 
the real parts, while it respects the $\mtr U(n_1)$ and 
 $\mtr U(n_2)$ in the numerator and 
  the two factors $\mtr U(1)$ of each $\C^\times$ in the denominator. We conclude that
\[\gauge(\A,J)\cong  \frac{
\mtr U(n_1)\times \mtr U(n_2) }
{\mtr U(1)\times \mtr U(1) }\cong \mtr{PU}(n_1)\times \mtr{PU}(n_2)
\,. \qedhere
\]  
\end{proof*}

\subsection{Unimodularity and the gauge group}
It turns out that for real algebras the gauge group does not
automatically include the unimodularity condition, and this property
needs to be added by hand.  Since this is relevant for the algebra
that one uses as input to derive the Standard Model (cf. discussion in
\cite[Ch. 8.1.1, Ch. 11.2]{WvSbook}) we address also the unimodularity
of the gauge group, when the base itself is noncommutative.
\par 
Given a matrix representation $\varrho$ of a unital $*$-algebra $A$,
the special unitary group of $A$ is defined by
\[\mathcal{SU} (A ) := \{ m \in \mathcal U(A)\mid \det [\varrho(m)] =1 \} \,.\]
We now define the following morphisms $\delta_i : \mtr{GL}(n_i) \to \C^\times$,
 \begin{align}
 \delta_1 (g_1)= [\textstyle\det_{n_1}(g_1)]^{n_2} \qquad \text{ and }  \qquad \delta_2(g_2) = [\det_{n_2}(g_2)]^{-n_1}\,,
 \label{morfismos}
 \end{align}
 which shall be useful in the description of the special unitary group
 we care about (notice that both morphisms depend on the pair
 $(n_1,n_2)$ and the different signs in the exponents).

\begin{lemma}\label{thm:SU}The special unitary group of $A_1\otimes A_2$,
\[
\mtc{SU}(A_1\otimes A_2) = \{ u_1\otimes u_2 \in \mtc{U}(A_1\otimes A_2) \mid \det(u_1\otimes u_2)=1 \}
\]
fits in a short exact sequence of groups:
\[
1 \to \mtr{U}(1) \hookrightarrow \mtr U(n_1) \times_{\det} \mtr U(n_2 ) 
\stackrel{\kappa}\to \mtc{SU}(A_1\otimes A_2)  \to 1\,, 
\] where  $\mtr U(n_1) \times_{\det} \mtr U(n_2 )  $ is the (categorical) pullback of any of the two morphisms \eqref{morfismos} along the remaining one. 
\end{lemma}

\begin{proof}
  Define the homomorphism $\kappa$ by
  $(u_1,u_2)\mapsto u_1\otimes u_2$. Suppose that
  $u_1 \otimes u_2 \in \ker \kappa$, so
  $\kappa(u_1,u_2) = u_1 \otimes u_2 = 1_{n_1}\otimes 1_{n_2}$.  This
  means that there exists a $\lambda \in \C^\times $ with
  $u_1= \lambda 1_{n_1} $ and $u_2 = \lambda\inv 1_{n_2}$, but by
  assumption $u_i \in \mtr{U}(n_i)$, so $\lambda \in \mtr U(1)$.  Thus
  the image of the inclusion
  $\mtr U(1) \hookrightarrow \mtr U(n_1) \times_{\det} \mtr U(n_2 ) $
  $\lambda \mapsto (\lambda 1_{n_1} ,\lambda\inv 1_{n_2})$ is the
  kernel of $\kappa$. \par
  The last map to the right is the determinant. If
  $u_1 \otimes u_2 \in \im \kappa$, then by definition of the fibered
  group $\mtr U(n_1) \times_{\det} \mtr U(n_2 )$,
  $\delta_1 (u_1) = \delta_2(u_2)$ holds. But this happens if and only
  if
  $1=[\det_{n_1}(u_1)]^{n_2}\cdot [\det_{n_2}(u_2)]^{n_1}= \det
  (u_1\otimes u_2)=(\det \circ \kappa) (u_1,u_2)$. Therefore the image
  of $\kappa$ is in the kernel of the determinant. \par
  On the other hand, if
  $g_1\otimes g_2\in \ker (\det ) \subset \mtc{SU}(A_1\otimes A_2) $
  then each $g_i \in \mtr{GL}(n_i) $ (otherwise its determinant
  vanishes and by assumption it is $1$) so we can write them in matrix
  polar form $g_i= p_i u_i$ with $u_i \in \mtr{U}(n_i)$ and
  $p_i =p_i^*$ positive definite.  Since, in particular,
  $p_1 u_1\otimes p_2u_2\in \mtc{U}(A_1\otimes A_2)$, one obtains
\begin{equation}
1_{n_1} \otimes 1_{n_2}= p_1 u_1 u_1^* p_1^* \otimes p_2 u_2 u_2^* p_2^* = p_1^2 \otimes p_2^2\,. \label{quitarcuadrados}
\end{equation}
Being both $p_i$'s positive definite Hermitian matrices, they can be
written as $p_i= v_i \Lambda_i v_i^*$ for
$\Lambda_i =\diag (\lambda_{i,1},\ldots, \lambda_{i,n})$ with
$\lambda_{i,m} \geq 0$ and $v_i \in \mtr U (n_i)$. But then
\eeqref{quitarcuadrados} means the existence of certain $r\in \re^+$
for which $v_1 (\Lambda_1)^2 v_1^* = r \cdot 1_{n_1} $ and
$v_2(\Lambda_2)^2 v_2^* = r\inv \cdot 1_{n_2}$.  Solving each equation
leads to $\Lambda_1=r^{1/2}1_{n_2}$ and $\Lambda_2 = r^{-1/2}1_{n_2}$,
so we can forget the $v_i$'s, since $\Lambda_i$ is central.

In summary, there exist scalars $\rho_i$ such that 
$p_i=\rho_i 1_{n_i}$ with $\rho_i>0$ and $\rho _1=1/\rho_2$. This relation shows
that $g_1 \otimes g_2 =   u_1\otimes u_2 = \kappa(u_1,u_2)$,
since in the tensor product $ g_1 \otimes g_2 = ( \lambda\inv g_1 ) \otimes (\lambda g_2)$
for any $\lambda\in \C^\times$ (here, in particular, choosing $\lambda= \rho_1$).
By construction, $u_i$ are unitarities, which, by assumption, moreover satisfy $1=\det(g_1\otimes g_2)
=\det( u_1 \otimes u_2 )=\delta_1(u_1)  / \delta_2(u_2)$. Hence $(u_1,u_2)\in \mtr U(n_1) \times_{\det} \mtr U(n_2 )  $
and $g_1\otimes g_2 = \kappa(u_1,u_2)$, which concludes the proof of exactness at $\mtc{SU}(A_1\otimes A_2)$.
\end{proof}

\begin{lemma}\label{thm:secondSU}
The following sequence of groups is exact:
\[
1 \to \mu_{\mtr{mcd}(n_1,n_2)} \stackrel{\iota}\to \mtr{U}(1) \times \mtr{SU}(n_1) \times \mtr{SU}(n_2) 
 \stackrel{\xi}\to \mtr{U}(n_1)\times_{\det}   \mtr{U}(n_2) \stackrel{\zeta}\to \mu_{n_1\cdot n_2} \to 1
\]
where $\mtr{mcd}(n_1,n_2)$ is the maximum common divisor of $n_1 $ and $n_2$.
\end{lemma}

\begin{proof}
From left to right we start defining the maps and checking exactness along the way.
The first map is $\iota(\lambda) = (\lambda, \lambda\inv \cdot   1_{n_1}, \lambda \cdot 1_{n_2})$.
Since $\det_{n_i}(z\cdot 1_{n_i})=z^{n_i}=1$ for $z \in  \mu_{\mtr{mcd}(n_1,n_2)}$,
the map is well-defined, and clearly is also injective.
\par 
The next map is given by $ \xi(z,m_1,m_2) = (z m_1,z\inv m_2)$.  Since
$m_i$ have unit determinant, the condition
$\delta_1 (z m_1 ) = z = 1/ \delta_2(z\inv m_2) $ is satisfied
(cf. \eeqref{morfismos} above). The pair $(z m_1,z\inv m_2)$ is thus
in the fibered product $ \mtr{U}(n_1)\times_{\det} \mtr{U}(n_2)$, by
its definition and $\xi$ is thus well-defined. \par
To verify the exactness, notice that if $(z,m_1,m_2) $ is such that
$\xi(z,m_1,m_2)=(zm_1,z\inv m_2)=(1_{n_1},1_{n_2})$, since each
$m_i \in \mtr{SU}(n_i)$, one has $\det_{n_1}(z\inv 1_{n_1})=1$ and
$\det_{n_2}(z 1_{n_2})=1$. Hence $z \in \mu_{n_1}\cap \mu_{n_2}$.
Thus $(z,m_1,m_2)= (z,z\inv \cdot1_{n_1}, z\cdot 1_{n_1})$ and
therefore $\ker \xi \subset \im \iota$ since the group that generates
this intersection $z \in \mu_{n_1}\cap \mu_{n_2}$ is
$ \mu_{\mtr{mcd}(n_1,n_2)}$. The other containment holds also, since
 \[\xi \circ \iota(\lambda)= \xi (\lambda , \lambda\inv  \cdot  1_{n_1}, \lambda\cdot  1_{n_1})= 
   (\lambda \cdot [\lambda\inv \cdot 1_{n_1}], \lambda \inv \cdot
   [\lambda \cdot 1_{n_2}]\,)= (1_{n_1},1_{n_2})\] for each
 $\lambda \in \mu_{\mtr{mcd}(n_1,n_2)}$. Hence $\ker \xi = \im \iota$
 and the sequence is exact at the node having the triple product. \par
 The last map is given by
 $ \zeta(u_1,u_2)= \textstyle [\det_1(u_1)]^{1/n_1} \cdot
 [\det_2(u_2)]^{1/n_2}$, which for $(u_1,u_2)$ in the fibered group
 product satisfies, by definition,
 \[ 1=\delta_1(u_1) / \delta_2(u_2) = \{\textstyle
   [\det_1(u_1)]^{1/n_1} \cdot [\det_2(u_2)]^{1/n_2}\}^{n_1n_2}=
   [\zeta(u_1,u_2)]^{n_1n_2}\,.\] (well-definedness).  To see that
 $\ker \zeta \subset \im \xi$, take $(u_1,u_2)$ in the fibered product
 group satisfying
 $\zeta (u_1,u_2)= \textstyle [\det_1(u_1)]^{1/n_1} \cdot
 [\det_2(u_2)]^{1/n_2} =1$.  This means that
\begin{equation}\textstyle \label{lambdazero}
\lambda_0 : =[\det_{n_1} (u_1)]^{1/n_1} \qquad \text{and} \qquad 
\lambda_0\inv   =[\det_{n_2}(u_2)]^{1/n_2}
\end{equation}
are consistent. Due to \eeqref{lambdazero}, conveniently used, both
matrices $\lambda_0\inv \cdot u_1 $ and $ \lambda_0 \cdot u_2 $ are
special unitary, and we also obtain
$(u_1,u_2)=\xi(\lambda_0, \lambda_0\inv \cdot u_1, \lambda_0 \cdot u_2
)$.\par
Finally, on the other hand,
\[(\zeta\circ \xi)(\lambda, m_1,m_2)= \zeta(\lambda m_1,\lambda\inv m_2) =\textstyle[\det_{n_1} (\lambda 1_{n_1})]^{1/n_1}
 [\det_{n_2} (\lambda 1_{n_2})]^{1/n_2}=1\] so
the inverted injection holds $\ker \zeta \supset \im \xi$ too.
\end{proof}

Lemma \ref{thm:secondSU} extracts the Lie group part of $\mtr U(n_1) \times_{\det} \mtr U(n_2 ) $. (This group appears in the description of the unimodular gauge group in Lemma \ref{thm:SU}.)  Its proof was inspired by one of 
Chamseddine-Connes-Marcolli, but is different from it due to the
presence of the tensor product of algebras, whilst
\cite[Prop. 2.16]{CCM} or \cite[Prop. 1.185]{ConnesMarcolli} focus on
unitarities of semi-simple algebras,
$A_1\oplus A_2\oplus \ldots \oplus A_k$.
\par 
In particular for the Standard Model
\cite[Prop. 1.199]{ConnesMarcolli}, the unimodular gauge group is the
well-known
$\{ \mtr U(1) \times\mtr {SU}(2) \times \mtr {SU}(3) \}/ \mu_6$
Standard Model gauge group (cf. also \cite[\S
6.2]{vandenDungen:2012ky}).  The embedding of the group $\mu_6$ of
sixth roots of unit in the Lie group is given by
$\lambda \mapsto (\lambda,\lambda^3,\lambda^2)$, as pointed out in
\cite[\S 11.2.1]{WvSbook}. Our embedding of the roots of unit
appearing in the above Lemma is different, since the determinant for
tensor products of algebras is governed by another rule:
$\det(a_1\otimes a_2)= [\det_{n_1}(a_1) ]^{n_2}
\times[\det_{n_2}(a_1)]^{n_1} $ for each
$ a_1\otimes a_2 \in A_1\otimes A_2$. On the physical side, the origin
of the two roots of unit groups in the exact sequence
\[1\to \mu_3 \to \mtr U(1) \times\mtr {SU}(2) \times \mtr {SU}(3) \to
  \mtr {SU}(A_F) \to \mu_{12} \to 1
  \quad\text{\cite[Seq. 1.661]{ConnesMarcolli}}\] characterizing the
unimodular gauge group for the algebra of the Standard Model
$A_F=\C\oplus \mtb H\oplus M_3(\C)$ is quite different: on the one
hand, the group\footnote{Here the fact that the unitary quaternions
  $\{ q\in \mtb H : q^*q=1=qq^* \}$ are unimodular (i.e. their
  determinant is 1 in the embedding of $\mtb H$ into $2\times 2$
  matrices), $\mtc U (\mtb H)\cong \mtr{SU}(2)$, causes that
  unimodularity has influence on the $\mtb H$ summand. That is why
  $\mu_3$ appears as fiber instead of $\mu_{3\times 2}$.} $\mu_3$
comes from $M_3(\C)$; and on the other $\mu_{12}$ does depend also on
the number of generations and the representation of fermions.
\\

By way of contrast, an important one conceptually, we stress that for
$\mtr{SU}(n)$-Yang-Mills(--Higgs) finite geometries where one has
$A_1=\MN $ and $A_2 =\Mn$ (so $n_1=N$ and $n_2=n$ above), $n$ is the
`color' analogue, the two (special) unitary factors in Proposition
\ref{thm:gaugegroup} or the unimodular analogue above, have a
different nature. The $\mtr{PU}(N)$ [resp. $\mtr{SU}(N)$] describes
the symmetry of the base (and could be understood as the finite
dimensional analogue of diffeomorphisms of a manifold) and
$\mtr{PU}(n)$ [resp. $\mtr{SU}(n)$] along the fibers.

\section{Yang-Mills--Higgs theory with finite-dimensional
  algebras} \label{sec:YMH} The Higgs field being considered at the
same footing with the gauge bosons is one of the appealing
characteristics that is offered by the gauge theory treatment with
NCG. We now recompute the results of Section
\ref{sec:GeneralFinACGeom}, revoking the restriction $D_F=0$. The aim
is a formula informed by Weitzenb\"ock's.  The Weitzenb\"ock formula,
$D_\omega^2= \Delta^{\mathbb S\otimes E}+\mathcal E $, includes the
Higgs $\Phi$ and extends Lichnerowicz's formula, to the product of the
spinor bundle $\mathbb S$ with a vector bundle $E$. It is given in terms of an
endomorphism $\mathcal E$ in $\Gamma(\mathrm{End}(\mathbb S \otimes E))$:
\begin{align}
\mathcal E= \frac14 R\otimes 1 + 1\otimes \Phi^2 - \sum_{i,j} \frac12 \ii \Gamma^i \Gamma^j \otimes \mathbb F_{ij} +\sum_{j} \ii \gamma_M \Gamma^j\otimes \adj(\nabla_j^{\mathbb S\otimes E}) \Phi\,,
\end{align}
where $(\nabla^{\mtb S\otimes E})_j$ and $\mathbb F_{ij}$ are locally
the connection on $\mathbb S \otimes E$ and the curvature on $E$,
respectively.  Further, $\gamma_M$ is the chirality element or
$\gamma_5$ in physicists' speak. (See e.g. \cite[Prop. 8.6]{WvSbook}
for a proof.)

\subsection{The Higgs matrix field}
We now turn off the fuzzy-gauge part of the spectral triple in order
to compute the fluctuations along the finite geometry $F$. These
fluctuations are namely generated by the second summand in the
original (in the sense, `unfluctuated') Dirac operator of the product
spectral triple
$D= D\fuz\otimes 1_F + \gamma\fuz \otimes D_F = D\fuz\otimes 1_F +
\gamma \otimes 1_{\MN} \otimes D_F$ where
$D_F=D^*_F \in \Mn_{\mtr{s.a}}$ is the Dirac operator of the finite
geometry $F$.

\begin{proposition}\label{thm:fluct_Higgs}
 The inner fluctuations of the Dirac operator 
 along the finite geometry $F$ are
 \begin{align}
 (\omega_{F} + J\omega_{F}  J\inv)(\Psi)= 
 (\gamma\otimes   \phi) (\Psi )  + \epsilon'' (\gamma\otimes 1_{\MN} \otimes 1_{\Mn}) \Psi(  1_V\otimes \phi )\,,
 \end{align}
 for each $\Psi\in \H=V\otimes \MN\otimes \A_F$. 
 These are parametrized by $\phi\in  \MN\otimes \Omega^1_{D_F}(\Mn)$.
 Also $\phi^*=\phi$ holds. 
\end{proposition}

\begin{proof}
As before, one computes the corresponding Connes' 1-forms $
 \ac [\gamma\fuz \otimes D_F, \cc ]  $ 
in terms of $\ac = 1_V\otimes W\otimes a$ and $\cc= 1_V\otimes T \otimes c $,
being  $W,T\in \MN$ and $a,c\in \Mn$. Namely,%
%
%
\begin{salign}
\omega_{F} & = \ac [\gamma\fuz\otimes D_F, \cc ]\\&  = \ac [\gamma \otimes 1_{\MN}\otimes D_F, \cc ] 
\\ & = (1_V\otimes  W \otimes  a) \big [ \gamma \otimes 1_{\MN} \otimes D_F, 1_V \otimes T \otimes c \big ]  \\ 
&= \gamma\otimes W T \otimes a [D_F,c ] 
\end{salign}
We rename $\phi:= X \otimes a [D_F,c ] $, since $W,T$
are arbitrary and their product can replaced by any matrix $X\in
\MN$. Thus
$\omega_{F}=\gamma\otimes \phi \in \Omega^1_{\gamma\fuz\otimes
  {D_F}}(\A)= \MN\otimes \Omega^1_{{D_F}}(\Mn)$ as claimed.  Since
from the onset $\gamma$ is self-adjoint, so must be $\phi$, since
$\omega_{F}^*=\omega_{F}$ is required. The remaining part of the
fluctuations acting on
$v\otimes Y\otimes m \in V \otimes \MN\otimes \Mn$ are
\begin{salign}
  (J \omega_{F} J\inv ) (v\otimes Y\otimes m)& = 
  \big( (C\otimes  *_N\otimes *_n ) ( \gamma\otimes \phi ) (C\inv \otimes *_N\otimes
  *_n ) \big)(v\otimes Y\otimes m ) \\ & = (C\otimes *_N \otimes *_n )
  (\gamma C\inv v \otimes X Y^*\otimes a [D_F,c ] m^*) \\ & = C \gamma C\inv
  v \otimes Y X^*\otimes (a [D_F,c ] m^* )^*\label{irrelevante}
  \numerada \\ & = \epsilon'' \gamma v \otimes Y X^*\otimes m (a
  [D_F,c ])^* \,,
\end{salign} 
since  $C \gamma  = \epsilon'' \gamma C$ (cf. table of Def. \ref{def:fuzzy}).
Therefore,
\begin{align}\nonumber
(J \omega_{F} J\inv ) (\Psi) & = \epsilon''\{\gamma\otimes 1_{\MN}\otimes 1_{\Mn}\}  (\Psi)
 \{ (1_V\otimes X^* \otimes (a [D_F,c ] )^* \}
 \\ &=  \epsilon'' \{\gamma\otimes 1_{\MN}\otimes 1_{\Mn}\}  (\Psi)
 ( 1_V\otimes \phi )\,, 
\end{align}
since $\phi= X \otimes a [D_F,c ] $ is self-adjoint, as argued before.  
\end{proof}
In \eeqref{irrelevante} of the proof one could also have computed 
directly, using the explicit formula \eqref{gammadef}
for the chirality: 
\begin{salign} C\gamma C\inv  &= (C \sigma(\eta) \gamma^0 C\inv ) (C \gamma^1 C\inv)( C \gamma^2 C\inv )( C \gamma^3   C\inv ) \\
&= (C  {\sigma(\eta)}   \gamma^0 C\inv)  \gamma^1\gamma^2\gamma^3     =  \overline{\sigma(\eta)}   \gamma^0\gamma^1\gamma^2\gamma^3  = \pm \gamma  
 \,. 
\end{salign} 
The complex conjugate in the last line appears since $C$ is
anti-linear. The sign is chosen as follows: notice that
$ \overline{\sigma(\eta)} $ is purely imaginary for the (1,3) and
(3,1) signatures (and otherwise it is a sign). This means that the
sign $\pm$ in last equation is
$(-1)^{\#\text{number of minus signs in $\eta$}}=(-1)^q $.  This
different way to compute leads to the same result as the one given in
the proof. Indeed, for 4-dimensional geometries $(-1)^q$ is precisely
$\epsilon''$, according to the sign table in Definition
\ref{def:fuzzy}, namely $\epsilon''=-1$ for KO-dimensions $2$ and $6$
and and $\epsilon''=+1$ for KO-dimensions $0$ and $4$.  \\

From Proposition \ref{thm:fluct_Higgs} and Theorem \ref{thm:flucutatedD_AC}, the 
form of the most general fluctuated Dirac operator follows: \vspace{-2pt}
\begin{subequations}\label{Diracfullyfluctuated}
 \begin{align}
 D_\omega& =  \overbrace{ \vphantom{\sum} \gamma\otimes \Phi }^{\Dh} +\sum_{\mu}\overbrace{\gamma^\mu \otimes (\km_{\mu} + \am_{\mu}) + \gahmu \otimes (\xm_{\mu}+\sm_{\mu}) }^{\Dg}   \,,\\
 D & = \gamma \otimes   D_F+\sum_{\mu}\gamma^\mu  \otimes \km_\mu + \gahmu  \otimes \xm_\mu    \,,\\[2pt]
 \omega & = \gamma \otimes   \phi+ \sum_{\mu} \gamma^\mu  \otimes \mathsf A_\mu + \gahmu  \otimes \mathsf S_\mu \,, \\[3pt]
 J\omega J\inv & =  \epsilon'' \gamma \otimes (\balita )  \phi  + \sum_{\mu}e_\mu  \gamma^\mu  \otimes (\balita )  \mathsf A_\mu + (-1)^q e_{\mu}\gahmu  \otimes (\balita ) \mathsf S_\mu \,,  \\
  \text{with }\Phi :\!\!&= 1_{\MN}\otimes D_F+\phi + \epsilon'' (\balita) \phi =  1_{\MN}\otimes D_F+\{ \phi, \balita \}_{\epsilon''} \vphantom{\sum_a}\,. \label{Higgsdef}
\end{align}
\end{subequations}%
%
We will call
$\Phi \in \MN\sa \otimes \{\Omega^1_{D_F}[\Mn]\}\sa \subset \big\{
M_N[\Omega^1_{D_F}(\Mn)]\big\}_{\mtr{s.a.}}$ the \textit{Higgs field},
since in the smooth Riemannian case (where the analogous relation
reads $\Phi^{(C^\infty)}=D_F + J_F \phi^{(C^\infty)} J_F\inv$) its analogue in the context of
almost-commutative geometries leads to the Standard Model Higgs field,
when the finite algebra $\A_F$ is correctly chosen
(cf. \cite{CCM,WvSbook}).  
%

\begin{corollary}\label{thm:flatWeitzenbHiggs}
The fluctuated Dirac operator $D_{\omega}$ on the `flat' ($\xm=0$) fuzzy space 
factor $G_\fay$ of a gauge matrix geometry $G_\fay \times  F$ satisfies
\begin{align}\label{flatWeitzenbHiggs}
D_{\omega}^2|_{X=0}=  
  \frac12 \sum_{\mu,\nu} \gamma^\mu\gamma^\nu \otimes   \mathscr{F}_{\mu\nu}+ 1_V\otimes (\vartheta + \Phi^2) +\sum_\mu \gamma^\mu \gamma \otimes [\day_\mu,\Phi] \,.
\end{align}

\end{corollary}

\begin{proof}
$ (D_\omega)^2= \Dg^2 + \Dh^2 + \{\Dg,\Dh\} $.  The gauge part $\Dg^2 $ is known
from Proposition \ref{thm:flatWeitzenb}; on the other hand,
$\Dh^2= (\gamma\otimes \Phi)^2=1_V\otimes \Phi^2$ from the axiom in
Definition \ref{def:fuzzy} for the chirality $\gamma$. 
Finally, $\{\Dg,\Dh \}= \sum_\mu \gamma^\mu \gamma \otimes [\day_\mu,\Phi]$,
since $\gamu \gamma = - \gamma \gamu$.
\end{proof}

Notice
also that $\gahmu=\gamma^\alpha \gamma^\rho \gamma^\sigma$
anti-commutes with $\gamma$, for 
\[\gamma^\alpha \gamma^\rho\gamma^\sigma \gamma = 
  - \gamma^\alpha \gamma^\rho \gamma \gamma^\sigma = + \gamma^\alpha
  \gamma \gamma^\rho \gamma^\sigma = - \gamma \gamma^\alpha
  \gamma^\rho\gamma^\sigma\,. \] 
Since the matrices $\gamma^\mu$ and
$\gahmu, \mu\in \Delta_4,$ span (the projection to $V$ of) $\Dg$, the
anti-commutator $\{\Dg,\Dh \}$ is traceless also if the fuzzy space is
`curved', $X\neq 0$.  
\begin{table} 
\includegraphics[width=\textwidth]{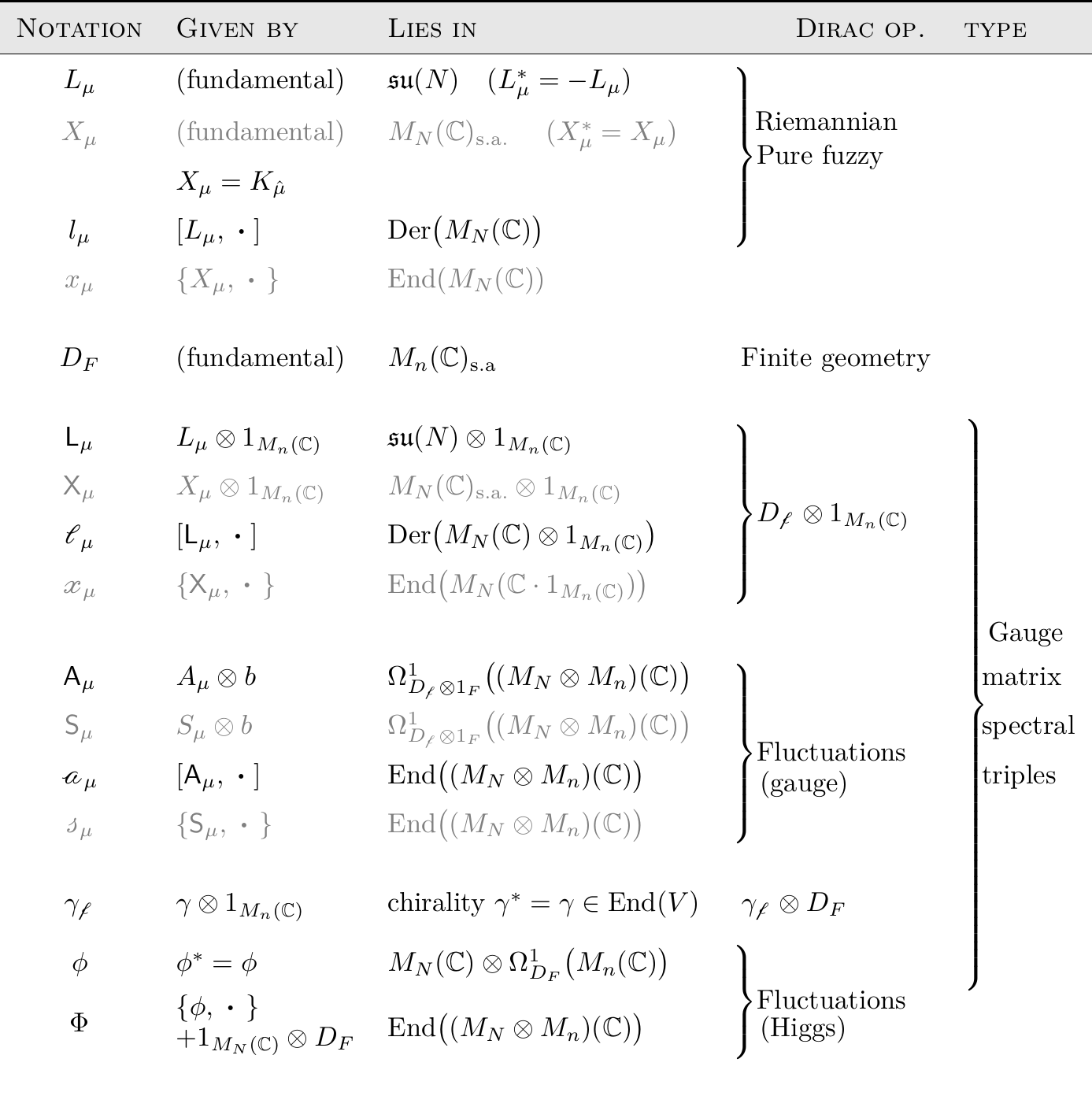}
 \caption{
 Notation for the matrices parametrizing the Dirac operator of \textit{Riemannian}
 four-dimensional Yang-Mills--Higgs matrix spectral triples
 and its fluctuations along $D=D\fuz\otimes 1_{F} + \gamma\fuz\otimes D_F$, which are split into blocks along the gauge ($D\fuz\otimes 1_{F} $) and Higgs parts ($\gamma\fuz\otimes D_F$). The accompanying gamma-matrices in the former case
 are omitted. 
 The rows in gray will not be used below ($X$ is set to zero, and this implies 
 the vanishing of the rest of operators in gray rows). 
 See \eeqref{Diracfullyfluctuated} for more details.
 \label{tab:notationRiemanniancase}}
\end{table}

\subsection{Transformations of the matrix gauge and Higgs fields}

Throughout this section, we always assume the Riemannian signature. 
We now compute the effect of the gauge transformations, already
explicitly known for the Dirac operator, on the field strength
$\mathscr F_{\mu\nu}$ and on the Higgs field. For the former, this requires to know how the matrices
$\mathsf{A}_\mu$ transform under
$\gauge(\A;J)=\mathcal U(\A)/ \mathcal U(\A_J)$.  We can pick a
representing element of $\gauge(\A;J)$ in $u\in \mathcal U(\A)$
directly, since the apparent ambiguity up to an element
$z\in \mathcal U(\A_J)$ leads to\footnote{The next equation is based on van Suijlekom's \cite[\S 8.2.1]{WvSbook}}
\begin{align*}
\omega ^{uz}& =(uz)\omega (uz)^* + uz [D,(uz)^*]  \\ &= u \omega u^* + uz \big \{ [D,z^*]u^*+ z^*[D,u^*] \big\}\\ &= u \omega u^* + uzz^*[D,u^*]= \omega ^{u}\,.
\end{align*}
The last line is obtained since  $\mathcal U(\A_J)= \mathcal U( \mathcal Z(\A))$,
so $z$ is central  (and thus $z^*$ too). Hence $[D,z^*]=0$.

\par 
Next, observe that, by definition, and also by Jacobi identity on $\MNn$, 
\begin{align} \label{forlateruse}
[\lm_\mu , \am_\nu]_\circ  & =
\big[\mathsf L_\mu,\comm{\mathsf A_\nu} \big]
-\big[\mathsf A_\nu,\comm{\mathsf L_\mu} \big] \\
&=\big[ [\mathsf L_\mu,\mathsf A_\nu],\balita \big] \nonumber
\end{align}
with analogous expressions for $[\lm_\nu, \am_\mu]_\circ $  and $[\am_\mu, \am_\nu]_\circ$.
This allows to write the field strength as the commutator with 
another quantity $\mathsf F_{\mu\nu} \in \MNn$ that we call \textit{field strength matrix},
\begin{subequations}
\begin{align}
 \mathscr F_{\mu\nu}&= \comm{ \mathsf F_{\mu\nu} }\,,\\
 \mathsf F_{\mu\nu} :\!&= 
 [\mathsf L_\mu + \mathsf A_\mu ,\mathsf L_\nu + \mathsf A_\nu] \,. \label{FieldStrenthMatrix}
\end{align}\end{subequations}%

We now find the way the field strength transforms under the gauge
group. By definition, the transformed field strength is given by the
expression $ \mathscr F_{\mu\nu}$ evaluated in the transformed
potential $\am_\mu^u$, this latter being dictated by the way the Dirac
operator transforms under $\gauge(\A;J)$.  Specifically,
 \begin{subequations}
   \begin{align} 
\mathscr F_{\mu\nu}^u & = [\lm_\mu + \am^u_\mu , 
\lm_\nu +\am^u_\nu]_{\circ} \,, \qquad   \\[2pt]
\am^u _\mu&= \comm{ \mathsf A_\mu ^u }\,, \\[2pt]
\mathsf A_\mu ^u &=  \Adj_u(\mathsf A_\mu ) + u[\mathsf L_\mu, u^*]
= u(\mathsf A_\mu )u^* + u[\mathsf L_\mu, u^*]\,,\\[2pt]
 \mathsf F_{\mu\nu}^u :\!&=
 [\mathsf L_\mu + \mathsf A^u_\mu ,\mathsf L_\nu + \mathsf A^u_\nu] ,.
\end{align}
Concerning the Higgs, we come back to Eqs. \eqref{Diracfullyfluctuated}.
We deduce from there and from \eqref{GaugeD}, that
the matrix field $\phi$, which parametrizes
by \eqref{Higgsdef} the Higgs field, transforms like
\begin{align}\phi\mapsto \phi^u = u \phi u^* + u[D_F,u^*]\, ,\qquad u\in\gauge(\A;J)\,. 
\end{align}
\end{subequations}
The transformation of the field strength  is more interesting: 
\begin{proposition}\label{thm:GaugeInvFs} 
In Riemannian signature, the field strength of a Yang-Mills(--Higgs) 
finite geometry transforms under the gauge group as follows:
 \begin{align}\label{field_transforms}
   \mathscr F_{\mu\nu}=\comm{\mathsf F_{\mu\nu} } \mapsto \mathscr F_{\mu\nu}^u= \comm{\mathsf F_{\mu\nu}^u} 
\qquad u\in \gauge(\A;J)\,,
\end{align}
which is completely determined by the next transformation rule on the
field strength matrix
\begin{align}
 \mathsf F_{\mu\nu} \mapsto  \mathsf F^u_{\mu\nu} & = 
 u
 \mathsf F_{\mu\nu} u^*   = \Adj_u(\mathsf F_{\mu\nu})  \,.
\end{align}

\end{proposition}
\begin{proof}
  Observe that for the pair $\lm_\mu, \am^u_\nu$ the same argument
  given about \eeqref{forlateruse} for the pair $\lm_\mu, \am_\nu$
  holds, and so does for the other pair of composition commutators
  appearing in the $u$-transformed field strength. Therefore, we can
  indeed write it in terms of the matrix
  $\mathsf F_{\mu\nu}^u:=[\mathsf L_\mu,\mathsf L_\nu] +[\mathsf L_\mu,\mathsf A_\nu^u ] - [\mathsf
  L_\nu,\mathsf A^u_\mu] + [\mathsf A_\mu^u,\mathsf A_\nu^u]$ as
  follows:
\begin{align}
\mathscr F_{\mu\nu}^u &= 
\big[ [\mathsf L_\mu,\mathsf L_\nu] + [\mathsf L_\mu,\mathsf A_\nu^u ] -  [\mathsf L_\nu,\mathsf A_\mu^u] + 
 [\mathsf A_\mu^u,\mathsf A_\nu^u]
,\balita \big] = \comm{\mathsf F^u_{\mu\nu}}\,.
\end{align}
We now compute the transformed field strength matrix
and all those terms that imply $\mathsf A$ (namely the
transformation under $u$ of $\mathsf T_{\mu\nu}:=
\mathsf F_{\mu\nu} -[\mathsf L_\mu,\mathsf L_\nu] \to 
\mathsf T_{\mu\nu}^u $) and infer from that those the gauge transformations on the 
field strength matrix $\mathsf F_{\mu\nu}$.
\begin{align*}
 \mathsf T^u_{\mu\nu}& = 
 +\big [\mathsf{L}_\mu, \Adj_u(\mathsf{A}_\nu) + u[\mathsf{L}_\nu,u^*] \big]
 -\big [\mathsf{L}_\nu, \Adj_u(\mathsf{A}_\mu) + u[\mathsf{L}_\mu,u^*] \big] 
 \\ & \quad +
 \Big[ \Adj_u(\mathsf{A}_\mu ) + u[\mathsf{L}_\mu,u^*],  \Adj_u(\mathsf{A}_\nu ) + u[\mathsf{L}_\nu,u^*] \Big] \\
 & = +\big[\mathsf{L}_\mu, u \mathsf{A}_\nu u^* \big ]+  \big [\mathsf{L}_\mu, u [\mathsf{L}_\nu, u^*] \big]
 \\
 &\quad - \big[\mathsf{L}_\nu, u \mathsf{A}_\mu u^* \big ] -  \big [\mathsf{L}_\nu, u [\mathsf{L}_\mu ,u^*] \big] \\
&\quad +  \big[ u\mathsf{A}_\mu u^*, u \mathsf{A}_\nu u^*\big] + \big[ u [\mathsf{L}_\mu, u^*], u\mathsf{A}_\nu u^*\big] \\
&\quad + \big[ u[\mathsf{L}_\mu,u^*], u[\mathsf{L}_\mu,u^*]\big]
+ \big[u\mathsf{A}_\mu u^*, u[\mathsf{L}_\nu,u^*] \big]
\end{align*} 
The contributions to $\mathsf T^u_{\mu\nu}$
split into three: 
$\mathsf{L}\mathsf{A}$-terms (i.e. containing $\mathsf L_\mu,\mathsf A_\nu$ or $\mathsf L_\nu ,\mathsf A_\mu$),
$\mathsf{A}\mathsf{A}$-terms,
and $\mathsf{L}\mathsf{L}$-terms. We compute them separately:
\begin{itemize}
 \itemb $\mathsf{A}\mathsf{A}$-terms: $ \big[ u\mathsf{A}_\mu u^*, u \mathsf{A}_\nu u^*\big]= u [\mathsf{A}_\mu,\mathsf{A}_\nu] u^*$, clearly
 \itemb $\mathsf{L}\mathsf{L}$-terms: When the commutators are expanded, the next $\mathsf{L}\mathsf{L}$-terms  
  \[ \big [\mathsf{L}_\mu, u [\mathsf{L}_\nu, u^*] \big] -  \big [\mathsf{L}_\nu, u [\mathsf{L}_\mu ,u^*] \big]
 + \big[ u[\mathsf{L}_\mu,u^*], u[\mathsf{L}_\mu,u^*]\big]\] yield the quantity 
 in bracelets, which can be neatly rewritten:
 \begin{salign}
 \left\{ \hspace{-12pt}
 \begin{array}{lr}
  &\phantom+ \mathsf{L}_\mu u \mathsf{L}_\nu u^*-\mathsf{L}_\mu \mathsf{L}_\nu -u \mathsf{L}_\nu u^* \mathsf{L}_\mu + \mathsf{L}_\nu \mathsf{L}_\nu\\
  &+\mathsf{L}_\nu \mathsf{L}_\mu - \mathsf{L}_\nu u \mathsf{L}_\mu u^* +u\mathsf{L}_\mu u^* \mathsf{L}_\nu - \mathsf{L}_\mu \mathsf{L}_\nu \\
  &+u \mathsf{L}_\mu \mathsf{L}_\nu u^* -u \mathsf{L}_\nu \mathsf{L}_\mu u^* +\mathsf{L}_\mu \mathsf{L}_\nu -\mathsf{L}_\nu \mathsf{L}_\mu \\
  &+u \mathsf{L}_\nu u^* \mathsf{L}_\mu - \mathsf{L}_\mu u \mathsf{L}_\nu u^* + \mathsf{L}_\nu u \mathsf{L}_\mu u^* -u\mathsf{L}_\mu u^* \mathsf{L}_\nu
\end{array} \right\} =u[\mathsf{L}_\mu,\mathsf{L}_\nu] u^*-[\mathsf{L}_\mu,\mathsf{L}_\nu] 
\label{llterms}
\end{salign} 
 \itemb $\mathsf{L}\mathsf{A}$-terms: $\big[\mathsf{L}_\mu, u \mathsf{A}_\nu u^* \big ] - \big[\mathsf{L}_\nu, u \mathsf{A}_\mu u^* \big ]+\big[ u [\mathsf{L}_\mu, u^*], u\mathsf{A}_\nu u^*\big] +\big[u\mathsf{A}_\mu u^*, u[\mathsf{L}_\nu,u^*] \big]$. This can be also obtained expanding 
 the commutators as above; the last two commutators 
 yield $ u\{[\mathsf{L}_\mu,\mathsf{A}_\nu]-[\mathsf{L}_\nu,\mathsf{A}_\mu]\}u^* + r(\mathsf{L},\mathsf{A})$.
 The excess terms $r(\mathsf{L},\mathsf{A})$ are actually cancelled out with the two first propagators, yielding 
 for the final expression of the $\mathsf{L}\mathsf{A}$-terms:
 \begin{align}
  u\big([\mathsf{L}_\mu,\mathsf{A}_\nu]-[\mathsf{L}_\nu,\mathsf{A}_\mu]\big)u^*
 \end{align}
\end{itemize} 
In view of the last equalities, we can conclude
that 
\begin{align} \label{Ts}
 \mathsf T^u_{\mu\nu} & = u
 \mathsf T_{\mu\nu} u^* +  u [\mathsf L_\mu,\mathsf L_\nu]
 u^*- [\mathsf L_\mu,\mathsf L_\nu]
 \\ & = \Adj_u(\mathsf T_{\mu\nu}) +
  \Adj_u \big([\mathsf L_\mu,\mathsf L_\nu]\big)-
  [\mathsf L_\mu,\mathsf L_\nu] \,, \nonumber 
\end{align} 
which, re-expressed in terms of $\mathsf F$, yields $\mathsf F_{\mu\nu}\to \mathsf F_{\mu\nu}^u = \Adj_u (\mathsf F_{\mu\nu})$.
\end{proof}

\begin{remark}\label{rem:Ls}
Notice that $\mathsf L_\mu$  being the fuzzy analogue of the
derivatives, the `surprising term' $[\mathsf L_\mu,\mathsf L_\nu]$  is the analogue\footnote{The precise 
statement is that 
$\adj_{[\mathsf L_\mu,\mathsf L_\nu]}$
is the analogue of $[\partial_\mu, \partial_\nu]$, but Jacobi identity used as in \eeqref{forlateruse}
allows one to state this in terms of $[\mathsf L_\mu,\mathsf L_\nu]$ only.} of 
$[\partial_\mu, \partial_\nu]$, which is identically zero on the algebra $C^\infty(M)$. This 
seems to (but, as we will see, does not) imply the freedom of choice 
as to whether we take the field strength matrix as 
defined above by ${\mathsf F}_{\mu\nu} $, or rather  $ \tilde {\mathsf F}_{\mu\nu} = 
{\mathsf F}_{\mu\nu} -[\mathsf L_\mu , \mathsf L_\nu]$ (called 
${\mathsf T}_{\mu\nu} $  above). 
According to \eeqref{Ts}, $\tilde {\mathsf F}_{\mu\nu}$ transforms then as 
 \begin{align}
 \tilde {\mathsf F}_{\mu\nu} \mapsto  \tilde{\mathsf F}^u_{\mu\nu} & =   \Adj_u(\mathsf F_{\mu\nu}) +
 \underbrace{ \Adj_u \big([\mathsf L_\mu,\mathsf L_\nu]\big)-
  [\mathsf L_\mu,\mathsf L_\nu]}_{\text{traceless}}\,.
\end{align}
Although for quadratic actions the last two terms add up 
to a traceless quantity, higher powers of the Dirac operator 
would mix the gauge sector with others. This confirms 
that the definitions in \eeqref{FieldStrength} and 
\eeqref{FieldStrenthMatrix} are correct. 
For only then, the pure gauge sector (i.e. powers of $\mathscr F$)
obtained from $\Tr _\H  (D^{2m}_{\omega\fuz}) $
would be expressible (see \cite{SAfuzzy}, and for $m=2$, \eeqref{tr_4gammas} above) 
as a sum over chord diagrams $\xi$, with $\boldsymbol \mu = (\mu_1,\ldots,\mu_m)$, $\boldsymbol \nu = (\nu_1,\ldots, \nu_m)$, 
\begin{align}
\sum_{\substack{\xi \\\,\,m\text{-chord diag.}}}\sum_ {\boldsymbol \mu,\boldsymbol \nu } \xi^{\mu_1 \nu_1 \mu_2 \nu_2 \ldots  \mu_m\nu_m}  \TrMNn \{ \mathscr   F^u_{\mu_1\nu_1} \cdots \mathscr  F^u_{\mu_m\nu_m} \}\,.\end{align}
The scalars 
$\xi^{\mu_1 \nu_1 \mu_2 \nu_2 \ldots  \mu_m\nu_m}$ 
are expressed as sums of $m$-fold products of the bilinear
form  $\eta^{\alpha\sigma}$ (signature) and are irrelevant for the discussion. 
The important conclusion is that, due to Proposition \ref{thm:GaugeInvFs}, the traced quantity is gauge invariant, since the transformation rule 
ignores the `space-time indices' $\mu_i$ and $\nu_i$. 
The quartic computation is explicitly given below. 
\end{remark}

\subsection{Traces of powers of $D$}

 The next statement is obvious:

\begin{lemma}\label{thm:TrD2}
  The fully fluctuated Dirac operator on the Yang-Mills--Higgs 
  matrix spectral triple satisfies in `flat space' (i.e. $X=0$),
\begin{align}\label{TrD2}
  \frac14 \TrH\big(D_{\omega}^2\big|_{X=0} \big) 
 & = \TrMNn ( \vartheta + \Phi^2 )
\end{align}

\end{lemma}
\begin{proof}First, $\{\Dh,\Dg\}=\gamma^\mu \gamma \otimes [\day_\mu,\Phi]$,
is traceless by $\{\gamma,\gamma^\mu\}=0$.
Now, from \eeqref{flatWeitzenbHiggs}, 
since tracing the first summand yields,
due to index symmetry of $\eta$ and index skew-symmetry of 
the field strength,
\[\sum_{\mu,\nu}\TrH(\gamu\ganu\otimes \mathscr{F}_{\mu\nu} )=\sum_{\mu,\nu} \eta^\munu\dim V
\TrMNn \mathscr{F}_{\mu\nu} = 0\]
one gets the result by \eeqref{Theta}. 
\end{proof}

\begin{lemma}\label{thm:TrD4}
  The fully fluctuated Dirac operator of the Yang-Mills-Higgs finite geometry  satisfies in `flat space' (i.e. $X=0$),
\begin{align} \nonumber
\frac{1}{4}\TrH \big(D_{\omega}^4\big|_{X=0} \big) 
&=   -\frac12 \sum_{\mu,\nu} \TrMNn (\Fcurv_\munu \Fcurv^\munu)  + \TrMNn \big( (\vartheta + \Phi^2)^2 \big)
\\&-\sum_{\mu,\nu}\eta^{\mu\nu} \TrMNn \big ( [\day_\mu, \Phi][\day_\nu,\Phi] \big)
\label{TrD4}\,.
\end{align}

\end{lemma}

\begin{proof}
Squaring the expression 
for $(D_\omega|_{X=0})^2$ given by 
Lemma \ref{thm:flatWeitzenbHiggs}. 
\begin{align*}
(D_\omega |_{X=0}^2)^2 & =
\frac{1}{4} \sum_{\mu,\nu,\rho,\sigma}\gamma^\mu \gamma^\nu \gamma^\rho \gamma^\sigma 
\otimes \mathscr F_{\mu \nu} \mathscr F_{\rho\sigma} \\& \vphantom{\frac12}+1_V \otimes (\vartheta + \Phi^2)^2
+ \sum_{\mu,\nu} (\gamma^\mu \gamma \gamma^\nu \gamma)\otimes [\day_\mu,\Phi] [\day_\nu,\Phi] \\ 
 & + \frac 12 \sum_{\mu,\nu}
\gamma^\mu \gamma^\nu \otimes \big (\mathscr F_{\mu \nu} ( \vartheta + \Phi^2
)+ ( \vartheta + \Phi^2
)\mathscr F_{\mu \nu} \big )
 \\
& +\vphantom{\frac12 \sum_\mu}\gamma^\mu\gamma \otimes\big( (\theta + \Phi^2) [\day_\mu,\Phi]
+ (\theta + \Phi^2) [\day_\mu,\Phi] \big) \\ &+ 
\frac12 \sum_{\rho,\mu,\nu} \gamma^\mu \gamma^\nu \gamma^\rho \gamma\otimes \big( \Fcurv_{\mu\nu} [\day_\rho,\Phi] \pm [\day_\rho,\Phi] \Fcurv_{\mu\nu} \big) \,.
\end{align*}
for some sign $\pm$ in last line, which is in fact irrelevant
since $\TrV( \gamma^\mu \gamma^\nu \gamma^\rho \gamma)=0$ for any 
choice of indices. The line before the last is also traceless. Further, using 
$\TrV(\gamma^\mu\gamma^\nu \gamma^\alpha \gamma^\rho)$ 
   given in \eeqref{tr_4gammas}, 
\begin{align}\nonumber
\Tr _\H (D_{\omega}^2|_{X=0}^2) 
&= 
\frac{\dim V}{4}\cdot\sum_{\mu,\nu,\rho,\sigma}  \bigg(
\raisebox{-.40\height}{\includegraphics[height=1.2cm]{PartitionsGammas4b.pdf}  } \!\! \bigg) 
\TrMNn\big(  \mathscr F_{\mu \nu} \mathscr F_{\alpha\rho} \big) 
\\ &+ \dim V \TrMNn \big ( (\vartheta+ \Phi^2)^2 \big)\nonumber
\\&+ \TrV(\gamma^\mu\gamma\gamma^\nu\gamma) \TrMNn \big( [\day_\mu,\Phi] [\day_\nu,\Phi] \big)
\nonumber 
\\
& + \dim V\sum_{\mu,\nu } \eta^{\mu\nu } \TrMNn ( \mathscr F_{\mu \nu} (\vartheta+ \Phi^2)) \label{jakiesrownanieAp}\,.
\end{align}
By symmetry of $\eta$ and skew-symmetry of $\mathscr F$,
the first chord diagram vanishes, and by the same token,
also the second line in \eeqref{jakiesrownanieAp}. 
The second chord diagram comes with a minus sign 
and, using the skew-symmetry $\mathscr F$, one can 
see that the third diagram yields the same contribution,
namely $\sum_{\mu,\nu,\rho,\sigma}(-\eta^{\nu\rho}\eta^{\mu\alpha})
\TrMNn(\mathscr F_{\mu \nu}\mathscr F_{\alpha \rho})$. 
For the second to last line, 
$ \TrV(\gamma^\mu\gamma\gamma^\nu\gamma)=-\eta^{\mu\nu}\TrV 1_V$.
Dividing the whole \eeqref{jakiesrownanieAp} by $\dim V=4$ and get the claim. %
\end{proof}

\section{The Spectral Action for Yang-Mills-Higgs matrix spectral triples: towards the continuum limit}  \label{sec:smoothlim} 

We now give the main statement and,  
after its proof, we compare it with \cite[\S 2]{Chamseddine:1996zu},
which derives from NCG the Yang-Mills--Higgs theory over a smooth manifold.
Since in differential geometry the Einstein summation convention is 
common, we restore it here (also in the fuzzy context) together with the
raising and lowering of indices with the constant
signature $\eta^{\mu\nu}=(\eta_{\mu\nu}\inv) $ and $\eta_{\mu\nu}$. Using the lemmata of previous 
sections, we can give a short proof to the main result:

\begin{thm}\label{thm:SA24}
For a Yang-Mills--Higgs matrix spectral triple
on a $4$-dimensional flat ($X=0$) Riemannian ($p=0$) fuzzy base, 
the Spectral Action for a real polynomial 
 $f(x)=  \frac{1}2 \sum_{i=1}^m a_i x^i  $    
 reads
 \begin{align} \label{SAintheorem}
 \frac14\TrH f(D) = S^\fay_\YM  + S^\fay_\Hi +  S^\fay_\gH   + 
 S^\fay _\vartheta +  \ldots,
 \end{align}
where each sector is defined as follows:
\begin{subequations}\label{HiggsPotdef}
 \begin{align}
   S^\fay_\YM(\lm, \am):\!&= - \frac{a_4}4 \TrMNn (\Fcurv_\munu \Fcurv^\munu)\,, \\
 S^\fay_\gH ( \lm,\am ,\Phi) :\hspace{-3pt}&= - a_4
\TrMNn \big ( \day_\mu \Phi \day^\mu \Phi \big) \,,
\\[3pt]
S^\fay _\Hi (\Phi) :\hspace{-3pt}& = \TrMNn f_{\mathrm{e}}(\Phi)\,, 
 \\[3pt]
 S_\vartheta^\fay (\lm, \am ):\hspace{-3pt}&= \TrMNn f_{\mathrm{e}}\big( {\vartheta^{1/2}} \big)  \,.
 \end{align}%
\end{subequations}%
 and the rest terms in the ellipsis 
 represents operators $
 \TrMNn [P]$ being 
 $P \in \C \langle \,\lm_\mu,\am_\mu \mid \mu =0,\ldots, 3 \rangle $ of order $\geq 5$. Further,
 $f_{\mathrm{e}}$ is the even part of the polynomial $f$ truncated to degree $<5$.
Moreover, one obtains positivity for each of the following
functionals, independently:
\[  S^\fay _\vartheta,\, S^\fay_\YM , \,  S^\fay_\Hi  \geq 0\,, \qquad \text{if $a_4\geq 0$} \,.\]
\end{thm}
\noindent
\begin{proof} Recall $D=\Dg+\Dh$. It is
obvious that $\TrH(D)=0$. The possible crossed-products
contributions to $\Tr (D^3)$ are
$\Tr(\Dg^2\Dh)$ and $ \Tr (\Dg \Dh^2)$.
The former vanishes because
 in spinor space $V$
we have to trace over $\gamma^\mu \gamma^\nu \gamma$, which
vanishes. Similarly, $\Dh^3$ is traceless
since $\gamma^3=\gamma$ is, and
and $\Dg^3$ vanishes by $\TrV(\gamma^\mu\gamma^\nu\gamma^\rho)=0$.
Thus odd powers of
$D$ are traceless, at least for degrees $<5$.
\par 
Hence inside 
the trace over $\H$, $f$ can be replaced by its even part $f_{\mathrm{e}}$.
Notice that by Lemmas \ref{thm:TrD2} and \ref{thm:TrD4}, then
\allowdisplaybreaks[2]
\begin{align*}\frac14\TrH \bigg[f(D)   -\sum_{r=6}^m\frac{a_r}2 D^r\bigg]& =
\frac14 \Tr_V\otimes \TrMNn \Big\{
\frac{1}{2} a_2 D^2+\frac{1}{2}  a_4 D^4
\Big\} \\ 
&= \TrMNn \Big\{ 
\frac{a_2}{2} (\vartheta + \Phi^2 ) +\frac{a_4}{2} ( (\vartheta + \Phi^2 ))^2
 \\ & \quad- \frac{a_4}{4}
\Fcurv_{\mu\nu }\Fcurv^{\mu\nu } -\frac{a_4}{2}\eta^{\mu\nu} [\day_\mu,\Phi] [\day_\nu,\Phi] \Big\}\\
&=
\TrMNn \Big\{ 
\frac{1}{2} \big(a_2\vartheta + a_4\vartheta^2 + a_2\Phi^2 + a_4  \Phi^4 \big)\\&\quad - \frac{a_4}{4}
\Fcurv_{\mu\nu }\Fcurv^{\mu\nu }  + {a_4}  \Big ( \Phi^2 \vartheta  -\frac12 [\day_\mu,\Phi] [\day^\mu,\Phi] \Big)
\Big\}
\,.
\end{align*}
Observe that by 
definition, \eeqref{Theta}, $\vartheta \Phi^2=\day_\mu\day^\mu \Phi^2$,
so the last term yields  by expansion of the commutators and cyclicity,
\begin{equation} \label{supermegawazne}
{a_4}  \TrMNn\Big ( \Phi^2 \vartheta  -\frac12 [\day_\mu,\Phi] [\day^\mu,\Phi] \Big)
=-a_4 \TrMNn \big ( \day_\mu \Phi \day^\mu \Phi  \big)\,.
\end{equation}

The result follows by inserting the definitions from \eeqref{HiggsPotdef}
and by observing that the trace of $D^6$ is a
noncommutative polynomial (which we do not determine) of homogeneous
degree 6 in the eight letters $\am$ and $\lm$; this is,
in the worse case, the rest term in  \eqref{SAintheorem}.\par 
Regarding positivity:
 First, notice that
$\am _\mu \am ^\mu = \am_\mu (e_\mu \am_\mu) = \am_\mu (\am_\mu)^*$ is
a positive operator, and that so is
$\vartheta \in \mtr{End}( \MN \otimes \H_F) $ by the same token,
$ \vartheta=\sum_\mu(\am+\km) _\mu (\am+\km)^\ast_\mu \geq 0\,$.
Thus  $f_{\mathrm{e}}(\sqrt\vartheta)$ is well-defined and its trace
positive, since $f_{\mathrm{e}}$ is by definition an even polynomial. \par Further 
relations like $[\kay_\mu,\am_\nu]^*= - e_\mu e_\nu [\km_\mu,\am_\nu]$,  
and similar ones for all the commutators defining the field strength,
lead to $\Fcurv_{\mu\nu}^* = -e_\mu e_\nu   \Fcurv_{\mu\nu}  $.  Since $\eta=\diag (e_0,\ldots, e_3)$, one obtains the positivity of the operator 
\begin{align}\label{positivityF}
 -\Fcurv_{\mu\nu} \Fcurv^{\mu\nu } =  \Fcurv_{\mu\nu} (-e_\mu e_\nu   \Fcurv_{\mu\nu} )
 =\Fcurv_{\mu\nu}(\Fcurv_{\mu\nu})^* \geq 0\,,\,\qquad\text{(no sum)}.
\end{align}
Therefore, also the positivity holds summing over $\mu,\nu$, which is
a positive multiple of $S^\fay_\YM$, whose positivity also follows.   Similarly, since $\Phi$ is self-adjoint, even powers of it
are positive, thus so is $S^\fay_\Hi $.
\end{proof}

We now comment on the interpretation of this result. For fuzzy geometries, the equivalent of integration over the manifold
is tracing operators $\MN \to \MN$. (At the risk of being redundant,
notice that the unit matrix in that space has trace $N^2$.) First, recall
that $\Phi$ is self-adjoint. We identify the Higgs field $H$ on a smooth, closed manifold $M$ with
$\Phi$, so the quartic part  $\int_M |H|^4 \mtr{vol}$ of the potential
for the Higgs is $\TrMNn( \Phi^4) $.  In the Riemannian case, in order
to address the gauge-Higgs sector\footnote{I thank 
the attentional, anonymous referee, who found in previous
versions an inconsistent match between the fuzzy and smooth cases. 
This led to find the by far more closer correspondence 
\eeqref{supermegawazne}, for which integration by parts (performed
in a previous version) is not even needed.}, notice that 
since $\Phi=\Phi^*$, if $a_4=1$,
\begin{align} \label{gHiggs}
S^\fay_\gH ( \lm,\am ,\Phi) & =  -\TrMNn \big[ (\lm_\mu+\am_\mu) \Phi (\lm^\mu+\am^\mu) \Phi \big]
\leftrightarrow
 - \int _M   
\mathbb D_\mu H (\overline{\mathbb D^\mu H}) \mtr{vol} \,.
\end{align} 
This  interpretation of $\day_\mu=\lm_\mu+\am_\mu$ as the
covariant derivative $\mathbb D_\mu=\partial_\mu + \mathbb A_\mu$ for
Yang-Mills connection, with the local gauge potential $\mathbb A _\mu$
absorbing the coupling constant (cf. Def. \ref{def:FieldStr} and Remark \ref{rem:Ls}).  
\par 
Next, notice that $\Fcurv_\munu$ is a matrix-version of the
$\mathrm{SU}(n)$-Yang-Mills  (local) curvature $\mathbb  F_{\mu\nu} $  for the action
$S_{\textsc{ym}}$. If $a_4=1$, one has the exact correspondence
\begin{subequations}\label{YMactionF_identification}
 \begin{align}\label{YMactionF} \raisetag{35pt}
   S^\fay_\YM(\lm, \am)\!&= -\frac14 \TrMNn ( \Fcurv_\munu \Fcurv^\munu)
   \\ & \updownarrow \nonumber
   \\ 
  S_\YM(\mathbb A)&=-\frac14 \int_M  \Tr_{\sun} (\mathbb  F_{\mu\nu}   {\mathbb F^{\mu\nu}}) \mtr{vol}\,.
  \end{align}
\end{subequations}%
  For the time being, the previous identifications hold only the
  Riemannian signature, since for $(p,q)\neq (0,4)$ anti-commutators
  appear; these, unlike commutators, are no longer derivations in the
  algebraic sense. Nevertheless, keeping this caveat in mind, we
  extend the previously defined functionals to any signature (there,
  each $\lm_\mu$ is replaced by $\km_\mu$). It holds then in general
  signature.

The identification of
$\mathscr{F}_{\mu\nu}= [\lm_\mu , \lm_\nu]_{\circ}+ [\lm_\mu , \am_\nu]_{\circ} - [\lm_\nu ,
\am_\mu]_{\circ} + [\am_\mu, \am_\nu]_{\circ} $ in the Riemannian case
(and its extension
$\mathscr{F}_{\mu\nu}= [\km_\mu , \km_\nu]_{\circ}+ [\km_\mu , \am_\nu]_{\circ} - [\km_\nu ,
\am_\mu]_{\circ} + [\am_\mu, \am_\nu]_{\circ} $ to the general
signature) with the curvature
$\mathbb F_{\mu\nu}=\partial _\mu \mathbb A_\nu -\partial_\nu \mathbb
A_\mu + [\mathbb A_\mu,\mathbb A_\nu]$ of the smooth case is further supported
by the fact that $[\lm_\mu , \am_\nu]_{\circ}$ generalizes the
multiplication operator $\partial_\mu \mathbb A_\nu$, on top of the 
reason already given in Remark \ref{rem:Ls}. The alternative
to this definition, using only $\lm_\mu \circ \am_\nu$ in place of $[\lm_\mu , \am_\nu]_\circ$ (and similar replacements), yields instead
$(\partial_{\mu} \circ \mathbb A_{\nu}) \psi= (\partial_\mu
\mathbb{A}_\nu) \cdot \psi +\mathbb{A}_\nu \partial_\mu (\psi) $ on
sections $\psi$ (fermions).  Notice also that for the smooth field
strength one gets the positivity of the type of \eeqref{positivityF}, namely
$-\Tr_{\sun} (\mathbb F_{\mu\nu} \overline{\mathbb F^{\mu\nu}}) \geq 0
$, due to $\mtb F_{\mu\nu} \overline{\mtb F^{\mu\nu}}=-\mtb F_{\mu\nu}\mtb
F_{\mu\nu} $ \cite[below Eq. 1.597]{ConnesMarcolli}. 
We summarize this section in Table \ref{tab:smoothanalogy}.

\begin{remark} \label{rem:tetrah}
Notice that in the expression for the Yang-Mills
action, when the model is fully expanded in terms of 
the fields $\km$ and $\am$, the next tetrahedral action appears
\begin{align} \label{tetrahedron_action}
 S^\fay\!\!\!\tetrah(\km) :\hspace{-3pt}&=- \frac12 \sum_{\mu\neq \nu} \TrMNn  \big( \km_\mu \km_\nu \km^\mu \km^\nu  \big)
\end{align}
as well as the same type of action, $S^\fay\!\!\!\tetrah(\am)$, in the variable $\am$.
The reference to a tetrahedron is justified when one writes that action in full, 
\[
  (\km_{\color{blue!60} \mu \color{black}})_{ij} (\km_{\color{blue!60}
    \nu \color{black}})_{jm}(\km^{\color{blue!60} \mu
    \color{black}})_{ml} (\km^{\color{blue!60}\nu \color{black}})_{li}
  \sim
  \raisebox{-.4195\height}{\includegraphics[height=3cm]{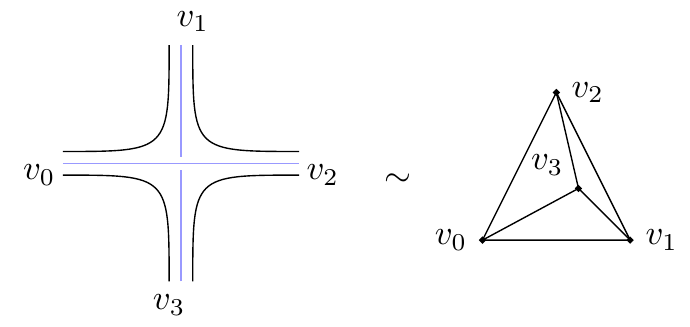}}\]
where the faint (blue) lines correspond to contractions of Greek
indices and black lines to matrix-indices $i,j,m,l$.  Modulo the
restriction $\mu \neq \nu$ present in the sum, this kind of action
$ S^\fay\!\!\!\tetrah$ is an example of the `matrix-tensor model'
class \cite{Benedetti:2020iyz}.
\end{remark}

\begin{table}[h] %
\includegraphics[width=\textwidth]{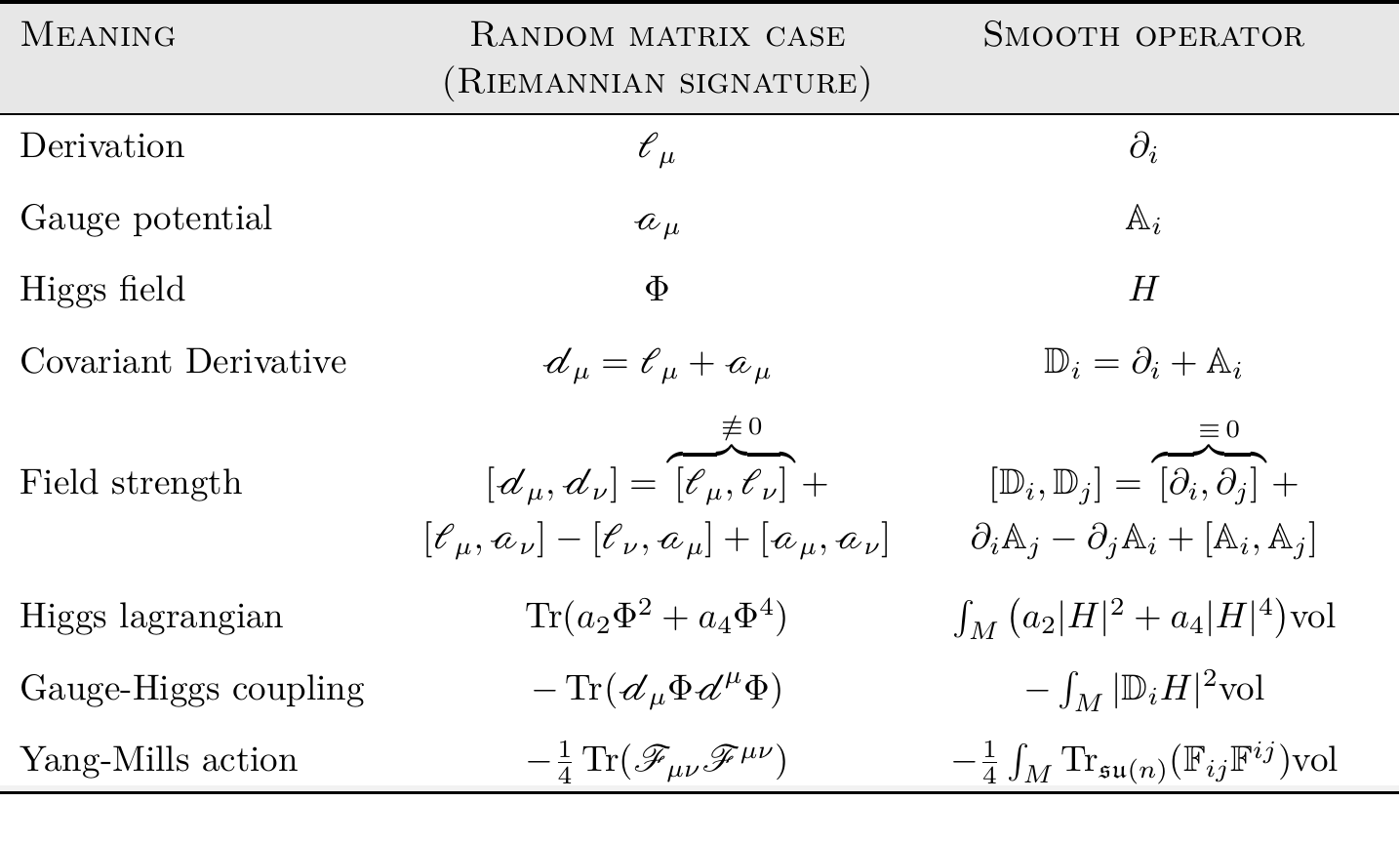}

 
 \caption{Only in this table, $\Tr \mathscr P$ denotes the trace of operators $\mathscr P:\MNn \to \MNn$;
\textit{gauge potential} means the local expression for the connection.
Finally, $a_2 $ and $a_4$ stand for for real parameters in $f$ 
in Theorem \ref{thm:SA24} which are particularly relevant for the Higgs Lagrangian, 
see \eeqref{HiggsPotdef}. The analogies implying $\lm_\mu \leftrightarrow \partial_i$ hold only for the Riemannian signature.
\label{tab:smoothanalogy}}
\end{table}



\section{Conclusions}\label{sec:Conclusions} 
We introduced gauge matrix spectral triples, computed
their spectral action and interpreted it as Yang-Mills--Higgs theory,
if the inner-space Dirac operator is non-trivial (and as Yang-Mills
theory if it is trivial), for the four dimensional geometry of
Riemannian signature.  We justified this terminology based on
Remark \ref{rem:Ls} and Section \ref{sec:smoothlim}; in particular see Table
\ref{tab:smoothanalogy} for the summary. The partition function of the
Yang-Mills-Higgs theory is an integral over gauge potentials
$\mathsf A_\mu$ and a Higgs field $\Phi $ in (subspaces of the)
following matrix spaces 
\[
\mathsf{A}_\mu  \in M_n(\Omega^1\fuz) \quad\text{and} \quad\phi \in M_N ( \Omega^1 _F)
\]
where $\Omega^1\fuz$ and $\Omega^1_F$ are the Connes' 1-forms along
the fuzzy and the finite geometry, respectively, both parametrized by
(finite) matrices, see Section \ref{sec:Outlook}. Additionally, the partition function for the spectral action implies an
integration over four copies of $\suN$; each of these matrix variables
$\mathsf{L}_\mu$ appears as the adjoint
$\lm_\mu=\adj_{\,\mathsf{L}_\mu}=\comm{\mathsf L_\mu} $. These
operators $\lm_\mu$ are interpreted as degrees of freedom solely of
the fuzzy geometry, in concordance with the identification of
$\mtr{Der}(\MN)$ with a finite version of the derivations on
$C^\infty (M)$, that is, vector fields. \par

As in the almost-commutative setting $M \times F$, with $M$ a
smooth manifold, the Higgs field arises from fluctuations along the
finite geometry $F$ and the Yang-Mills gauge fields from those along
the smooth manifold $M$. This is apparent in the parametrizing matrix
subspaces (see \eeqref{NgaugeNHiggs} below) for the matrix Higgs field and the
matrix gauge potentials, which are swapped if one
simultaneously\footnote{To be strict, one has to swap also the
  anti-Hermiticity by the Hermiticity in the both lines of
  \eqref{NgaugeNHiggs}, but this is clearly fixed by an imaginary
  factor and is ignored here.} exchanges $n \leftrightarrow N$ and
$F \leftrightarrow \fay$.  The Yang-Mills--Higgs matrix theory has a
projective gauge group $\gauge= \mtr{PU}(N) \times \mtr{PU}(n)$.  The
left factor corresponds with the symmetries of the fuzzy spacetime and
the right one with those of the `inner space' of the gauge theory (a
similar interpretation holds for the unimodular gauge groups in Lemmas
\ref{thm:SU} and \ref{thm:secondSU}), so the whole group $\gauge$
could be understood as $C^\infty(M,\mtr{SU}(n))$ after a truncation
has been imposed on $M$. A rigorous interpretation, e.g. in terms of spectral truncations \cite{ConnesWvS}, is still needed.

Another approach to reach a continuum limit resembling smooth spin
manifolds is the Functional Renormalization Group, which could be
helpful in searching the fixed points (cf. the companion paper
\cite{FRGEmultimatrix} for the application of this idea to general
multimatrix models).

\section{Outlook}\label{sec:Outlook}

Aiming at a model with room for gravitational degrees of freedom, the
careful construction of a Matrix Spin Geometry needs a separate study
(in particular requiring $X_\mu\neq 0$ and thus also a more general
treatment than that of Section \ref{sec:YMH}). If that is concluded,
one could identity for signature $(0,4)$

\begin{itemize}\setlength\itemsep{.2em}
 \itemb Lemma \ref{thm:Lichnerowicz} with `Fuzzy Lichnerowicz formula',
 \itemb Proposition \ref{thm:flatWeitzenb}
 with `Fuzzy flat Weitzenb\"ock formula', and 
 \itemb Proposition \ref{thm:Weitzenb} with `Fuzzy Weitzenb\"ock formula'.
\end{itemize}

In order to
give a more structured appearance to the partition function
for Riemannian ($p=0$), flat Yang-Mills--Higgs spectral triples, we recall
the dependence of our functionals on the \textit{fundamental matrix fields}
$\mathsf L_\mu,\mathsf A_\mu $ and $\phi$.
The $\mathsf L$'s are functioning as derivatives 
$\mathsf L_\mu \in \suN \otimes 1_n $, and
$\lm_\mu= \adj_{\,\mathsf L_\mu}= \comm{\mathsf L_\mu}$ is the
derivation defined by the adjoint action,
$\lm_\mu \in \mtr{Der} \,\big(\MN\big) \otimes 1_n$, for each $\mu$.
One arrives at a similar situation with the \textit{matrix gauge potentials}
\[
\mathsf A_\mu \in  \big\{\Omega^1_{D\fuz}[ \MN]\big\}_{\mtr{anti\text{-}Herm.}}  \otimes \Mn_{\mtr{s.a.}}  \subset \MN \otimes \Mn\,,
\]
where the subindex in the curly brackets restricts to anti-Hermitian
1-forms.  In terms of these
$ (\am_0,\am_1 ,\am_2 , \am_3) =\am = \am(\mathsf A_\mu) $ is defined,
again, via derivations:
$\am_\mu = \adj_{\, \mathsf A_\mu} = \comm{\mathsf A_\mu } $, which
already bear a non-trivial factor in the inner space\footnote{ We
  recall that $\am$ is actually dependent on $\lm$, since the
  operators $\am_\mu =\comm{\mathsf A_\mu}$ parametrize the inner
  fluctuations of the Dirac operator $D\fuz \otimes 1_F $, itself
  parametrized by $\lm$, but this dependence is not made
  explicit.}. This yields dependences $\lm= \lm(\mathsf L),\am=\am(\mathsf A)$.
  Further, by \eeqref{Higgsdef}, also $\Phi=\Phi(\phi)$. All in all, this yields for each sector
\begin{align*}
  S^\fay_\YM &=S^\fay_\YM (\mathsf L,\mathsf A)\,,  &&&   S^\fay_\gH&=S^\fay_\gH(\mathsf L, \mathsf A, \phi)\,,  
 \\ S^\fay_\Hi&= S^\fay_\Hi(\phi)\,,  &&&  
 S^\fay _\vartheta &= S^\fay _\vartheta( \mathsf L,\mathsf A)\,. 
\end{align*}

The partition function, using a polynomial $f(x) $, reads
\begin{align} \label{ZSA}
\mathcal Z^\fay = \int _{\mathscr N} 
\exp \Big(-\frac{1}{4} \TrH f(D) \Big) \dif D
\end{align}
where 
\begin{itemize}\setlength\itemsep{.4em}
 \itemb 
the Spectral Action is given by Theorem \ref{thm:SA24} 

\itemb the partition function $\mathcal Z^\fay=\mathcal Z_{N, n}^{\fay,f }$  
implies integration over the matrix space $\mathscr N$ that depends on the parameters $ N$ and $n$ 
via
\begin{align}\label{moduli}
\qquad
(\mathsf L_\mu,\mathsf A_\mu,\phi) \in 
\mathscr N = \mathscr N_{N,n}^{p=0,q=4} =\big[ \suN \big]^{\times 4}\times  \big[\mathscr N ^{\mtr{gauge}}_{N,n}\big]^{\times 4}
\times 
\mathscr N ^{\mtr{Higgs}}_{ N,n }  \,,
\end{align}
with the Higgs and gauge fields matrix spaces defined by
\begin{subequations} \label{NgaugeNHiggs}
\begin{align}
\mathscr N ^{\mtr{Higgs}}_{N,n}&:= 
 \ii \hspace{1pt}\mtf{u}(N) \otimes [\Omega^1 _{D_F} (\Mn )]_{\mtr{s.a.}} 
 \subset \big\{ M_N  \big[\Omega^1 _{D_F} \big(\Mn \big) \big]  \big\}_{\mtr{s.a.}}
\\
\qquad 
\mathscr N ^{\mtr{gauge}}_{N,n}&:= \ii \hspace{1pt} \Omega ^1 _{D\fuz}(\MN)_{\mtr{s.a.}} \otimes \ii \hspace{1pt} \mathfrak u (n) \subset \big\{ M_n\big[\Omega ^1 _{D\fuz}\big(\MN\big)\big] \big\}\antiH
\end{align}
\end{subequations}%
\itemb the measure $\dif D=\dif \mathsf L\, \dif \mathsf A \,\dif\phi  $
is the product of Lebesgue measures on the three factors of \eqref{moduli}.
\end{itemize}

Whilst writing down the path integral does not solve the general problem
of how to quantize noncommutative geometries, this finite-dimensional
setting might pave one of the possible ways there, for instance,
also by addressing these via computer simulations (Barrett-Glaser's aim). 
However, it
should be stressed that the treatment of this path integral is not yet
complete, due to the gauge redundancy to be still taken care of.  A
suitable approach is the BV-formalism\footnote{Tangentially, a
  discussion on gauge theories and the BV-formalism in the modern
  language of $L_\infty$-algebras appears in \cite{Ciric:2021rhi}, in
  a noncommutative field theory (but also different) context.} (after
Batalin-Vilkovisky \cite{Batalin:1984jr}), all the more considering
that it has been explored for $\mtr{U}(2)$-matrix models in
\cite{Iseppi:2016olv}, and lately also given in a spectral triple
description \cite{WomenMP}. Another direction would be
\cite{Nguyen:2021rsa}, \`a la Costello-Gwilliam.

En passant, notice that since the main algebra here is $ M_N(A)$ with
$ A$ a noncommutative algebra, the Dyson-Schwinger equations of these
multimatrix models would be `quantum' (in the sense of Mingo-Speicher
\cite[\S 4]{MingoSpeicher}; this is work in progress).

\footnotesize
\begin{acknowledgements} 
This work was mostly supported by the TEAM programme of the
  Foundation for Polish Science co-financed by the European Union
  under the European Regional Development Fund
  (POIR.04.04.00-00-5C55/17-00). Towards the end, this work was 
  funded by the European Research Council (ERC) under
the European Union’s 
Horizon 2020 research and innovation program (grant agreement
No818066) and also by the
Deutsche Forschungsgemeinschaft (DFG, German Research Foundation)
under Germany’s 
Excellence Strategy EXC-2181/1-390900948 (the Heidelberg \textsc{Structures}
Cluster of Excellence). I thank the Faculty of Physics, Astronomy and Applied Computer
  Science, Jagiellonian University for hospitality (in 2018). I am
  also indebted to Andrzej Sitarz and Alexander Schenkel for useful comments 
  during the preparation phase. 
\end{acknowledgements}
  
\appendix

\normalsize 
\section{Proofs of some lemmata}\label{app:prooflemma}

\begin{proof*}{Proof of Lemma \ref{thm:gamma13}}
  For the first equation: If $\mu=\nu$, then
  \[\gamma^\mu\gamma^{\hat \mu}=\gamma^\mu\gamma^0\cdots
    \widehat{\gamma^\mu} \cdots \gamma^3=(-1)^{\mu}\gamma^0\cdots
    \gamma^\mu \cdots \gamma^3\,,\] since the first $\gamma^\mu$ has
  to `jump' $\mu$ gamma-matrices in order to form
  $ \gamma^0\gamma^1\gamma^2\gamma^3$.  If $\mu \neq \nu$, precisely
  the two gamma matrices with indices different from $\mu$ and $\nu$
  survive, which explains $\delta_{\mu\nu\alpha\sigma}$.  The LHS of
  Eq. \eqref{gammas13} contains also $(\gamma^\mu)^2=e_\mu 1_V$.  To
  justify the sign $(-1)^\mu\mtr{sgn}(\nu-\mu)$, notice, by explicit
  computation, that if $\mu < \nu$ then the $(\mu,\nu)$-pairs
  $(0,\nu)$ and $(2,\nu)$ yield positive sign, whereas $(1,\nu)$
  negative.  Thus the sign is $(-1)^{\mu}$ (in no case the sign
  depends on $\nu$, as far as $\mu < \nu$).  The situation is inverted
  if $\mu> \nu$, where the pairs $(1,\nu)$ and $(3,\nu)$ yield
  positive sign and $(2,\nu)$ negative.  Thus the sign is
  $(-1)^{\mu-1}$.
 \par 
 Notice that in $\gamma^{\hat \mu}\gamma^\mu $ the last matrix has to
 move $4-\mu-1$ places to the right to arrive at the $\mu$-th factor,
 which says that
 $\gamma^{\hat \mu}\gamma^\mu=(-1)^{3-\mu} \gamma^0
 \gamma^1\gamma^2\gamma^3=-(-1)^{\mu} \gamma^0
 \gamma^1\gamma^2\gamma^3$ and, by \eeqref{gammas13},
 \eeqref{gammas13b} follows.\par
 For \eeqref{gammas13c} one notices that $\gamma^\mu$ has to jump, in
 order to pass to the other side, three matrices (one of which is
 $\gamma^\mu$ itself), which yields the sign $(-1)^2$.
 \par
 
 Regarding \eeqref{eq:gammastriples}
  To obtain the second summand, notice that if $\mu=\nu$, then the LHS
  is of the form $( \gamma^\iota \gamma^\lambda \gamma^\tau )^2 $ with
  pairwise different indices (i.e. anti-commuting gamma-matrices).
  Therefore
  $ \gamma^{\hat \mu }\gamma^{\hat \mu } =e_\iota \gamma^\lambda
  \gamma^\tau \gamma^\lambda \gamma^\tau =-e_\iota \gamma^\lambda
  \gamma^\lambda \gamma^\tau \gamma^\tau = - e_\iota e_\lambda e_\tau
  1_V$. Since each $e_{\!\!\balita}$ is a sign and
  $\{\iota,\lambda,\tau,\mu\}= \Delta_4$, by multiplying the last
  expression by $e_\mu^2=1$ one arrives to
  $\gamma^{\hat \mu }\gamma^{\hat \mu }=-e_\mu \cdot (e_0e_1 e_2 e_3 )
  1_V = -e_\mu \det(\eta)1_V $.
\par 
If $\nu\neq \mu$, we first determine the corresponding RHS term up to
a sign, and thereafter correct it. First, it is clear that
$\gamma^{\hat \mu }\gamma^{\hat \nu }$ is a product of $\gamma^\mu $
(which appears in $\gamma^{\hat \nu }$), with $\gamma^\nu$ (which
appears in $\gamma^{\hat \nu }$) and, additionally, with the other two
gamma-matrices whose indices $\{\lambda,\rho\}$ that are neither $\mu$
nor $\nu$. But each one of the latter appears twice, once in
$\gamma^{\hat \nu }$ and once in $\gamma^{\hat \mu }$.  The matrix is
then proportional to $\gamma^\mu \gamma^\nu$, which with the squared
matrices $\gamma^\lambda$ and $\gamma^\rho$ yield
$\varsigma_{\mu\nu} e_\lambda e_\rho \gamma^\mu\gamma^\nu$ for
$\lambda, \rho \in \Delta_4\setminus \{\mu,\nu\}$ and
$\lambda \neq \rho$, for a sign $\varsigma_{\mu\nu}=\pm $ that we now
determine. To enforce the inequality of all the indices we introduce
$\delta_{\mu\nu\lambda\rho}$, but since
$e_\lambda e_\rho \delta_{\mu\nu\lambda\rho}$ is symmetric in
$\lambda$ and $\rho$, we have to divide the sum over those indices by
$1/2$. To find the correct sign $\varsigma_{\mu\nu}$, by explicit
computation one sees that $\varsigma_{\mu\nu}=-1$ if and only if
$(\mu,\nu)$ is $(0,2), (2,0), (3,1)$ or $(1,3)$ and
$ \varsigma_{\mu\nu}=+1$ in all the other cases. That is,
$ \varsigma_{\mu\nu}=-1$ if and only if $|\mu-\nu|$ is even. But this
is precisely equivalent to $\varsigma_{\mu\nu}=(-1)^{|\mu-\nu|+1} $.
 \end{proof*}

\begin{lemma}\label{thm:CenterOfTensorProd} Let $A_i$ be unital, associative algebras, and let $Z(A_i)$
be the center of $A_i$. Then 
$   Z(A_1\otimes A_2) =   Z(A_1) \otimes   Z( A_2) $. 
\end{lemma}

\begin{proof}
  Notice that for (so far, arbitrary) $a_i, b_i \in A_i$ ($i=1,2$),
  one has by adding and subtracting $a_1b_1\otimes b_2a_2$ and
  rearranging,
\begin{equation}
 [a_1 \otimes a_2 , b_1\otimes b_2 ] = 
[a_1,b_1] \otimes (b_2 a_2 ) + (a_1 b_1) \otimes [a_2,b_2]\,.
\label{centrum}
\end{equation}
Clearly, if $a_i\in  Z(A_i)$ for $i=1,2$,
then the RHS vanishes for each $b_i \in A_i$, that is,  $a_1\otimes a_2\in   Z ( A_1 \otimes A_2)$. 
Therefore  $ Z(A_1)\otimes Z(A_2) \subset  Z ( A_1 \otimes A_2)$. \par Conversely,
notice that if $a_1\otimes a_2=0$, then we are done, so we suppose  
$a_1\otimes a_2 \in Z(A_1\otimes A_2)\setminus \{0\}$. If the LHS of the previous equation vanishes for  
each $b_1\otimes b_2 \in A_1\otimes A_2$, so does for $b_1=1$;
in which case, one gets $a_1 \otimes [a_2,b_2]=0$ for each $b_2 \in A_2$,
so $a_2\in Z(A_2)$, since $a_1\neq 0$ by assumption.  
Repeating the argument now taking $b_2=1$ instead, one gets $ Z(A_1)\otimes Z(A_2) \supset  Z ( A_1 \otimes A_2)$. \end{proof}

\begin{proof*}{Proof of Lemma \ref{thm:TrD4}} 
  Again, in the whole proof we set $K=0$, even though the notation
  will not reflect it. This can be obtained by small modifications
  from the previous lemma: if $K_\mu=0$ for each $\mu$, then
 \begin{salign}
\omega\fuz (\Psi)& = \ac [D\fuz\otimes 1_F, \ac' ] (\Psi)   = \sum_\mu\gamma^{\hat\mu} v \otimes W 
 [ X_\mu ,  T ]  Y \otimes  a c  \psi  \\ & = \sum_\mu\big( \gamma^{\hat \mu}\otimes W [X_\mu, T] \otimes a c \big) \Psi 
 \\& = \sum_\mu\big( \gamma^{\hat \mu} \otimes  S_\mu \otimes a c \big) \Psi 
\end{salign}%
 and $S_\mu:=W [X_\mu,T]$. Since $a$ and $c$ are arbitrary matrices,
 we rename $b=a c$. Again, since the 
\begin{align}
 \label{selfadjofGmu}
(\gamma^{\hat \mu} \otimes S_\mu \otimes b)^* &= e_{\hat \mu}
\gamma^{\hat \mu} \otimes   S_\mu ^* \otimes b^* 
 \end{align}
and since we had set already $b\in \ii \,\mathfrak{u}(n)$ 
 we conclude that $ S_\mu ^* = e_{\hat \mu}  S_\mu$. 
 This sign is  $e_{\hat \mu}=(-1)^{q+1}e_\mu$, according to \cite[App. A]{SAfuzzy}. 
 It follows from the definition, that the anti-linear operator $C:V \to V$ satisfies  
  $C(\gamma^\alpha \gamma^\rho \gamma^\sigma) C \inv = C \gamma^\alpha C\inv \cdot C \gamma^\rho  C\inv \cdot C \gamma^\sigma  C\inv
 =\gamma^\alpha \gamma^\rho \gamma^\sigma$ for each triple of indices $\alpha,\rho,\sigma\in \Delta_4$.
  Therefore $C\gamma^{\hat\mu} C\inv = \gamma^{\hat \mu} $ and the operator $J\omega\fuz  J\inv$
  can thus readily be computed: for $\Psi=v\otimes Y\otimes \psi \in V \otimes \MN\otimes \Mn$, 
  \allowdisplaybreaks[2]
  \begin{salign}
(J \omega\fuz  J\inv ) (\Psi)& = \vphantom{\sum_\mu}\big( J \ac [D\fuz\otimes 1_F, \ac' ] J\inv \big)  (\Psi)   \\ & = \sum_\mu (C\otimes *_N \otimes *_n )
(\gamma^{\hat \mu} C\inv v \otimes S_\mu Y^* \otimes b \psi ^*)
\\ & =
\sum_\mu  
(\underbrace{C \gamma^{\hat \mu} C\inv}_{\gamma^{\hat \mu}} v) \otimes (S_\mu Y^*)^* \otimes  (b \psi ^*)^* \\
&=\sum_\mu  
(\gamma^{\hat \mu} \otimes 1_{\MN}\otimes 1_{\Mn}) \Psi
(1_V\otimes e_{\hat \mu}  S_\mu \otimes b) \\ & = \sum_\mu  
( (-1)^{q+1}e_\mu \gamma^{\hat \mu} \otimes 1_{\MN}\otimes 1_{\Mn}) \Psi
(1_V\otimes S_\mu \otimes b) \,. \qedhere
\end{salign}
\end{proof*}


\end{document}